\documentclass[runningheads]{llncs}
\usepackage{OMMstyle}
\usepackage[T1]{fontenc}

\begin{document}
\title{Near-Optimal Online Metric Matching \\ on \(\Delta\)-ary HSTs}
%
%
\author{Parth Gor\inst{1}\orcidlink{0009-0005-1747-0564} \and 
Sourya Roy\inst{1}\orcidlink{0000-0001-9160-1503} \and 
Kasturi Varadarajan\inst{1}}
\authorrunning{P. Gor et al.}
%
\institute{University of Iowa, Iowa City, IA 52242, USA\\
\email{\{parth-gor, sourya-roy, kasturi-varadarajan\}@uiowa.edu}}
\maketitle              
\begin{abstract}
In the online metric matching problem, we have $n$ servers with known locations in some metric space. Requests arrive one-by-one at certain locations, and upon arrival a request must be matched to a server that was not matched to a previous request. The goal is to minimize the matching cost. For randomized algorithms with an oblivious adversary, the best known competitive ratio is obtained by embedding the metric space into an HST, and then solving the problem in the setting where the metric space is defined by the HST.

Bansal et al.~\cite{metriclog2gupta} introduced a framework for online metric matching where one can develop an algorithm in a {\em restricted reassignment model}, and then transform this into a true online algorithm. Using this framework, they obtained an expected competitive ratio of $O(\log n)$ for HSTs; this also gives the best known competitive ratio of $O(\log^2 n)$ for general metrics. In this paper, we revisit this framework with the aim of developing new algorithms. For HSTs where each node has at most $\Delta$ children, we develop an algorithm via this framework with an expected competitive ratio
of $O((\log\log \Delta) \cdot  \log \Delta)$. In particular, this ratio is independent of $n$, the number of servers/requests. It is near-optimal, as the expected competitive ratio of any algorithm is $\Omega(\log \Delta)$.

\end{abstract}

\section{Introduction}
In the \textit{Online Metric Matching} problem, we are given a set $\cS$ of $n$ servers with locations in some metric space $(X,d)$. A sequence $\cR$ of requests arrives one at a time, with each request located at some point in the metric space. When a request $r$ arrives, it can be matched to a server $\mu(r) \in \cS$ that was previously unmatched; this assignment cannot be changed subsequently. The objective function that needs to be minimized is the cost of matching all requests, defined as $\sum_{r \in \cR} d(r,\mu(r))$.
As an online problem, the performance of an algorithm is measured by its \textit{competitive ratio}, defined as the supremum, over all input instances, of the ratio between the total cost incurred by the online algorithm and the cost of an optimal offline solution ($\OPT$) that has full knowledge of the request sequence in advance. We say that an online algorithm is $f$-competitive if its competitive ratio is at most $f$.

The online metric matching problem was introduced, in independent work, by Kalyanasundaram and Pruhs~\cite{kalyanasundaram1993online} and Khuller et al.~\cite{KHULLER1994255}. These works gave a $(2n - 1)$-competitive deterministic algorithm for the problem, and showed that this competitive ratio is optimal for deterministic algorithms. 

Meyerson et al.~\cite{meyerson2006randomized} gave a randomized algorithm with an expected competitive ratio of $O(\log^3 n)$, assuming an oblivious adversary -- this means that the request sequence is unknown to our algorithm but is fixed; in particular, it does not depend on the random choices made by the algorithm. Unless explicitly stated otherwise, we assume this adversarial model in this paper in the context of randomized algorithms. The lower bound on the competitive ratio for randomized algorithms is $\Omega(\log n)$, which already holds in uniform metric spaces, where the distance between any two points is the same \cite{meyerson2006randomized}. Meyerson et al.~\cite{meyerson2006randomized} first probabilistically embed the input instance from the original metric space into an $\alpha$-hierarchically separated tree ($\alpha$-HST, see Section~\ref{sec:prelims}) that causes distances to be distorted in expectation by $\Theta(\alpha \log n)$ \cite{FRT_embedding03}. They present an algorithm, {\em the randomized greedy} strategy, that is $O(\log n)$-competitive on any $\alpha$-HST provided $\alpha = \Omega(\log n)$. 

Subsequently, Bansal et al.~\cite{metriclog2gupta} gave a randomized algorithm with an expected competitive ratio of $O(\log ^2 n)$. This is the best competitive ratio known so far, and bridging the gap between this upper bound and the lower bound of $\Omega(\log n)$ is a longstanding open problem. Bansal et al.~\cite{metriclog2gupta} embed the input instance into a $2$-HST, which only incurs an expected distance distortion of $\Theta(\log n)$. Then they give an $O(\log n)$-competitive algorithm on a 2-HST. In this paper, we further explore the framework that they introduce to obtain their optimal result on $2$-HSTs.

Gupta and Lewi~\cite{couplingLineGupta12} gave a randomized algorithm with expected competitive ratio $O(\log n)$ on metrics with constant doubling dimension, a class that includes Euclidean spaces of constant dimension. This result is also achieved via embedding onto HSTs. A $\Delta$-ary $\alpha$-HST is one where each internal node has at most $\Delta$ children. They present a {\em randomized subtree} strategy that achieves an $O(\log \Delta)$ competitive ratio on any $\Delta$-ary $\alpha$-HST provided $\alpha$ is sufficiently large. In particular, it is sufficient that $\alpha \geq \max \{c H_{\Delta}, 2\}$, where $c > 1$ is any constant, and $H_k = 1 + \tfrac{1}{2} + \cdots + \tfrac{1}{k}$ is the $k$'th Harmonic number. Note that the competitive ratio depends only on $\Delta$ for such HSTs. Gupta and Lewi \cite{couplingLineGupta12} show that any online matching instance in a metric of doubling dimension $\textsf{d}_{\dim}$ can be embedded into such a $\Delta$-ary $\alpha$-HST with $\alpha = \Theta(\textsf{d}_{\dim} \log \textsf{d}_{\dim})$ and $H_{\Delta} = \Theta(\textsf{d}_{\dim} \log \textsf{d}_{\dim})$. The expected distance distortion of the embedding is $O(\alpha \cdot \textsf{d}_{\dim} \cdot \log n)$.

Recently, Kalyanasundaram et al. \cite{KPS23} gave a randomized  online algorithm for the uniform metric that achieves (expected) competitive ratio of $\Theta(\log k)$, where $k$ is the number of distinct server locations in the metric. This algorithm, which guarantees a competitive ratio independent of the number $n$ of servers, is obtained via the reassignment paradigm. They also give a reduction from any instance on a 2-HST to instances on the uniform metric. They claim that this reduction gives an $O(\log k)$-competitive algorithm on any 2-HST with $k$ leaves, and thus an $O(\log^2 k)$ competitive ratio on an arbitrary metric, where $k$ is the number of distinct server locations.

The study of $\Delta$-ary $\alpha$-HSTs is particularly motivated by the problem of online matching in \textit{doubling metrics}. Here, probabilistic embedding onto an HST naturally yields trees with a branching factor bounded by the doubling dimension; see \cite{couplingLineGupta12} and the references therein. The recent work of Li, Vitercik, and Yang \cite{smoothed_analysis_delta-ary} also studies  online metric matching on $\Delta$-ary HSTs with $\Delta = O(1)$, 
under a setting where the requests arrive from a distribution that's close to uniform.
We also discuss some interesting related works in various models in Appendix \ref{sec:otherwork}.

\subsection{Our Contribution}
In giving an $O(\log n)$-competitive algorithm for online metric matching on a 2-HST, Bansal et al. \cite{metriclog2gupta} present a framework for online algorithms for metric matching. To explain our result, we describe a natural and straightforward modification of their framework. Here, one develops an algorithm for metric matching in an intermediate model called the \emph{Reassignment Model.} In the standard online version of the problem, matching decisions are immediate and irrevocable; requests cannot be reassigned. In the more-relaxed reassignment model, a request $r$ currently matched to a server
$s$ may be reassigned to another server $s'$, but we pay a cost for this reassignment. In this model, the cost incurred by the algorithm is the sum of all assignment and reassignment costs of requests. 
In the framework of Bansal et al.~\cite{metriclog2gupta}, the algorithm in the reassignment model is then converted into an algorithm in the online model without an increase in cost. 

In this paper, we explore this reassignment framework for $\Delta$-ary $\alpha$-HSTs. While the reassignment framework is used in \cite{metriclog2gupta} to obtain the best known competitive ratio of $O(\log^2 n)$ for online metric matching via a randomized algorithm, the optimal deterministic algorithms of Kalyanasundaram and Pruhs~\cite{kalyanasundaram1993online} and Khuller et al.~\cite{KHULLER1994255} can also be understood as using this framework. This motivates us to ask if the reassignment framework can be used to
design an algorithm for
$\Delta$-ary $\alpha$-HSTs that achieves the best competitive ratio in terms of \(\Delta\). 
 Our main result gives an affirmative answer to this question:

\begin{theorem}\label{thm:main-result}
 There is a randomized online algorithm, using the reassignment framework, for metric matching on $\Delta$-ary $2$-HSTs that has an expected competitive ratio of $O((\log \Delta) \log \log \Delta)$.   
\end{theorem}

Our competitive ratio is near-optimal, as the $\Omega(\log \Delta)$ lower bound follows for an HST of height 1 from the lower bound on the uniform metric. 
We note that a weaker competitive ratio of $O((\log \Delta)^2 (\log \log \Delta)^2 )$
can be recovered by using a two stage approach where we first ``\textit{embed}'' the input HST into a ``\textit{coarser}'' HST and then run the algorithm of~\cite{couplingLineGupta12} on this coarser tree to generate the assignments. However, for this two-stage approach it does not seem possible to get a ratio that beats $\Omega((\log \Delta)^2)$. The approach of \cite{KPS23} also does not give us a competitive ratio that beats $\Omega((\log \Delta)^c)$, for any $c \geq 2$; in Appendix \ref{sec:KPS}, we explain the $c = 2$ case.

Establishing our main result required us to develop a new reassignment strategy, as the original reassignment algorithm has a lower bound of \(\Omega(\log n)\) (see Appendix \ref{sec:app-bbgn-lb}) on its competitive ratio.

Among other things, a key distinguishing feature of our proposed reassignment strategy compared to~\cite{metriclog2gupta} is that our algorithm places substantially greater emphasis on doing ``local'' reassignments. This allows us to have a more precise control on tracking the reassignment costs. This however has the negative effect of increasing the number of reassignments significantly. One of our core technical contribution is to carefully argue that the costs for the reassignments are distributed in such a way that there are only a small number of reassignments that cost a lot, and a heavy tail of reassignments that are cheap. Showing this happens in the reassignment setting brings up significant subtleties and technical challenges, as explained in Section \ref{sec:overview}.

Interestingly, the \textit{sampling distribution} that we utilize to heavily prioritize local reassignments has not been considered before.
In fact, the only reassignment sampling distribution that has been studied so far is the one in~\cite{metriclog2gupta}. Thus, our work explores and suggests new possibilities for how different reassignment distributions can be used to design enhanced matching algorithms in various settings. 

For any finite metric with $n$ points, we can efficiently compute a probabilistic embedding into $2$-HSTs so that the expected distance distortion is $O(\log n)$ \cite{FRT_embedding03}. For a metric with constant doubling dimension, it is known that we can compute a probabilistic embedding into $2$-HSTs so that (i) the expected distance distortion is $O(\log n)$, and (ii) the degree of every tree node is $O(1)$ (\cite{GLR03}, Section 6). 

Consequently, Theorem~\ref{thm:main-result} implies an $O(\log n)$-competitive algorithm for online metric matching in metric spaces of constant doubling dimension, matching the bound obtained by Gupta and Lewi~\cite{couplingLineGupta12}. Thus, we show that this major result in the area can also be obtained using the reassignment framework.

Lastly, for an $\alpha$-HST with $k$ distinct server locations, the setting considered by \cite{KPS23}, our framework yields an optimal expected competitive ratio of $O(\log k)$, as outlined in Appendix~\ref{sec:leaves}.
\section{Notation and Preliminaries} \label{sec:prelims}

\def\gap{\textsf{G}}

Given a parameter $\alpha \geq 1$, an {\em $\alpha$-Hierarchically Separated Tree} ($\alpha$-HST) is a rooted tree along with a length function $d$ defined on the edge set of the tree, that satisfies the following conditions:

\begin{enumerate}
    \item For any node $v$, and any two children $v_1,$ $v_2$ of $v$, we have $d(v, v_1) = d(v,v_2)$.
    \item For any node $v$ (that is not the root), let $p(v)$ denote its parent; for any child $v_1$ of $v$, we have $d(p(v), v) = \alpha \cdot d(v, v_1)$. 
    \item The number of edges on every root to leaf path is the same; we call this number the height of the tree.
\end{enumerate}

For any two nodes $u$ and $v$ in the tree, we let $d(u,v)$ denote the sum of the lengths of the edges on the unique simple path between $u$ and $v$. The function $d(\cdot, \cdot)$ defines a metric on the nodes of the HST. From the definition it follows that for any node $v$ in the HST and any two leaves $v_1$ and $v_2$ of the subtree $T_v$ rooted at $v$, we have $d(v,v_1) = d(v, v_2)$. 

We define the {\em height} $h(v)$ of any node $v$ in the HST to be the number of edges on the path from $v$ to any leaf in subtree $T_v$. Thus, the leaves of the HST have height $0$. Let $\treeheight$ denote the height of the root node of the HST.

Without loss of generality, we can assume that the edge from any leaf to its parent has length $1$, as shown in Fig. \ref{fig:HST}. An HST is said to be {\em $\Delta$-ary} if each internal node has at most $\Delta$ children. We assume throughout that the parameter $\Delta$ is lower bounded by a suitable absolute constant $c$; this is reasonable because, for instance a $2$-ary HST is also a $c$-ary HST, and treating it as such simplifies certain expressions.

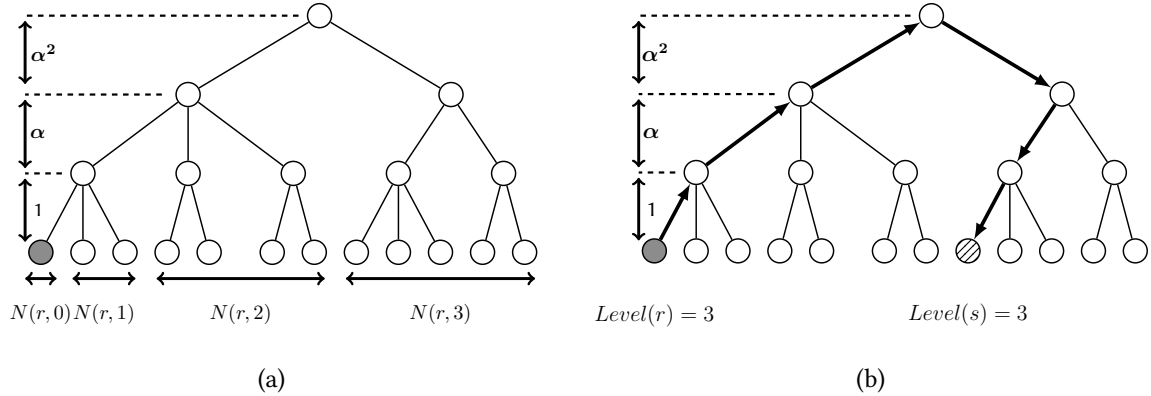
\begin{figure}[!htbp]
    \centering
    \begin{subfigure}[t]{0.47\textwidth}
        \centering
        \resizebox{\linewidth}{!}{
        \begin{tikzpicture}[
                baseline=(root.center),
                every node/.style={circle, thick, draw, minimum size=13pt, inner sep=0pt},
                labelnode/.style={draw=none, circle=none, inner sep=1pt, font=\bfseries\normalfont},
                edge from parent/.style={draw, thick},
                edge from parent path={(\tikzparentnode) -- (\tikzchildnode)}, 
                level 1/.style={sibling distance=5cm},
                level 2/.style={sibling distance=2cm},
                level 3/.style={sibling distance=0.8cm}
            ]
            
              \node (root) {}
                child { node (L1) {}
                  child {  node (L2) {}
                    child { node[fill=gray!90]{} }
                    child { node {} }
                    child { node {} }
                  }
                  child { node {}
                    child { node {} }
                    child { node {} }
                  }
                  child { node {}
                    child { node {} }
                    child { node {} }
                  }
                }
                child { node {}
                  child { node {}
                    child { node {} }
                    child { node {} }
                    child { node {} }
                  }
                  child { node {}
                    child { node {} }
                    child { node {} }
                  }
                };
            
                \draw[dashed, very thick] (-0.4, 0) -- (-5.6, 0);
                \draw[dashed, very thick] (-3, -1.5)   -- (-5.6, -1.5);
                \draw[dashed, very thick] (-4.9, -3)   -- (-5.6, -3);

                \draw[<->, ultra thick] (-5.6,-5) -- (-5,-5) node[midway, below, labelnode] {$N(r,0)$};
                \draw[<->, ultra thick] (-4.7,-5) -- (-3.5,-5) node[midway, below, labelnode] {$N(r,1)$};
                \draw[<->, ultra thick] (-3.1,-5) -- (0.1,-5) node[midway, below, labelnode] {$N(r,2)$};
                \draw[<->, ultra thick] (0.5,-5) -- (4.1,-5) node[midway, below, labelnode] {$N(r,3)$};

                \draw[<->, ultra thick] (-5.6,-0.1) -- (-5.6,-1.3) node[midway, right, labelnode] {$\boldsymbol{\alpha^2}$};
                \draw[<->, ultra thick] (-5.6,-1.6) -- (-5.6,-2.9) node[midway, right, labelnode] {$\boldsymbol{\alpha}$};
                \draw[<->, ultra thick] (-5.6,-3.1) -- (-5.6,-4.3) node[midway, right, labelnode] {\sf 1};
        \end{tikzpicture}
        }
        \caption{}
        \label{fig:image_a}
    \end{subfigure} %
    \hfill
    \begin{subfigure}[t]{0.49\textwidth}
        \centering
        \resizebox{\linewidth}{!}{
        \begin{tikzpicture}[
            baseline=(root.center),
            every node/.style={circle, thick, draw, minimum size=13pt, inner sep=0pt},
            labelnode/.style={draw=none, circle=none, inner sep=0pt, font=\bfseries\large},
            edge from parent/.style={draw, thick},
            edge from parent path={(\tikzparentnode) -- (\tikzchildnode)},
            level 1/.style={sibling distance=5cm},
            level 2/.style={sibling distance=2cm},
            level 3/.style={sibling distance=0.8cm}
        ]
        
          \node (root) {}
            child { node {}
              child { node {}
                child { node[fill=gray!90](request) {} edge from parent[draw, line width=2pt, latex-]}
                child { node {} }
                child { node {} }
                edge from parent[draw, line width=2pt, latex-]
              }
              child { node {}
                child { node {} }
                child { node {} }
              }
              child { node {}
                child { node {} }
                child { node {} }
              }
              edge from parent[draw, line width=2pt, latex-]
            }
            child { node {}
              child { node {}
                child { node[pattern=north east lines](server) {}  edge from parent[draw, line width=2pt, -latex] }
                child { node {} }
                child { node {} }
                edge from parent[draw, line width=2pt, -latex] 
              }
              child { node {}
                child { node {} }
                child { node {} }
              }
              edge from parent[draw, line width=2pt, -latex] 
            };
            
            \draw[dashed, very thick] (-0.4, 0) -- (-5.6, 0);
            \draw[dashed, very thick] (-3, -1.5)   -- (-5.6, -1.5);
            \draw[dashed, very thick] (-4.9, -3)   -- (-5.6, -3);
            
            \node[labelnode, below=-5.5pt of request] {$Level(r)=3$};
            \node[labelnode, below=-5.5pt of server] {$Level(s)=3$};

            \draw[<->, ultra thick] (-5.6,-0.1) -- (-5.6,-1.3) node[midway, right, labelnode]{$\boldsymbol{\alpha^2}$};
            \draw[<->, ultra thick] (-5.6,-1.6) -- (-5.6,-2.9) node[midway, right, labelnode] {$\boldsymbol{\alpha}$};
            \draw[<->, ultra thick] (-5.6,-3.1) -- (-5.6,-4.3) node[midway, right, labelnode] {\sf 1};
            
        \end{tikzpicture}
        }
        \vspace{-0.4cm}
        \caption{}
        \label{fig:image_b}
    \end{subfigure}
    \caption{(a) shows an $\alpha$-HST. For a leaf node $r$ (shaded in gray), it illustrates the set of nodes $N(r, \cdot)$. If request $r$ (shaded in gray) is matched to server $s$ (with lines) as shown in (b), then $Level(r) = Level(s) = 3$. (This figure has been adapted from~\cite{metriclog2gupta}).}
    \label{fig:HST}
\end{figure}

In the metric matching problem on HSTs, each server resides at some leaf of the HST.  Each request also arrives at some leaf. We allow multiple servers and requests at the same leaf. When we reference the location of a server or request in the HST, this is to be understood as the corresponding leaf node.

Following \cite{metriclog2gupta}, we introduce two important notions. For any leaf node $x$, and an integer $\ell \geq 0$ that is at most the height of the HST, 
\begin{itemize}
    \item $T(x,\ell)$ denotes the set of leaves of the subtree at height $\ell$ containing leaf $x$.
    \item Let $N(x, \ell) = T(x, \ell) \setminus T(x, \ell - 1)$ if $\ell > 0$ and let $N(x, 0) = T(x, 0)$. 
\end{itemize}

Note that if $y \in N(x, \ell)$, and $v$ is the least common ancestor of $x$ and $y$, then $h(v) = \ell$ and
\[ d(v, x) = d(v, y) = \sum_{j = 0}^{\ell - 1} \alpha^j\]
 
\paragraph{Special layers.}
     For integer $j \geq 0$, let $\SL_j$ denote the set of nodes in the HST at height $j * \gap$, where $\gap: = 2 \log \log \Delta$. (We round up if $2 \log \log \Delta$ is not an integer.) For any such $j$, we refer to $\SL_j$ as a {\em special layer} in the HST. Suppose that the special layers in the HST are $\SL_0, \SL_1, \ldots, \SL_k$; that is, $k$ is the largest integer such that $k * \gap \leq \treeheight$. Note that consecutive special layers are apart in height by $\gap = 2 \log \log \Delta$. The set of leaves in the tree is $\SL_0$, and it is possible that the root is not in a special layer.

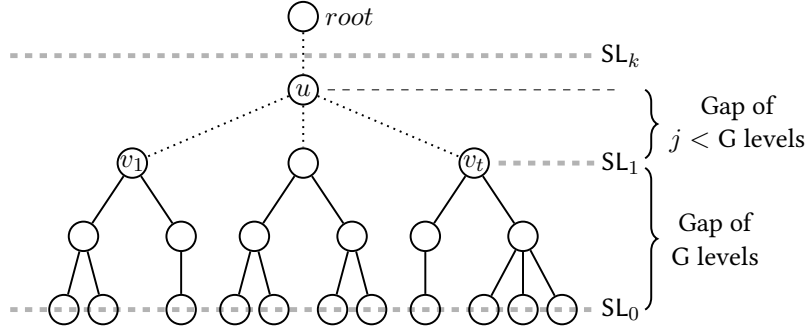
\begin{figure}[!htbp]
\centering
\resizebox{0.8\linewidth}{!}{\begin{tikzpicture}[
                every node/.style={circle, thick, draw, minimum size=12pt, inner sep=0pt},
                labelnode/.style={draw=none, circle=none, inner sep=1pt, font=\bfseries\normalfont},
                edge from parent/.style={draw, thick},
                edge from parent path={(\tikzparentnode) -- (\tikzchildnode)}, 
                level 1/.style={sibling distance=8cm},
                level 2/.style={sibling distance=3.5cm},
                level 3/.style={sibling distance=2cm},
                level 4/.style={sibling distance=0.8cm},
                scale=0.7
            ]
            
              \node (root) {}
              child { node (u) {$u$}
                child { node (v1) {$v_1$} 
                  child {  node (L2) {} 
                    child { node (r){} }
                    child { node {} }
                  }
                  child { node {} 
                    child { node {} }
                  }
                  edge from parent[draw=none]
                } 
                child { node (v2) {}
                  child { node {}
                    child { node (rp){} }
                    child { node (s){} }
                  }
                  child { node {}
                    child { node (s_b){} }
                    child { node (s_c){} }
                  }
                  edge from parent[draw=none]
                }
                child { node (vt){$v_t$}
                  child { node {}
                    child { node (s_d){} }
                  }
                  child { node {}
                    child { node {} }
                    child { node {} } 
                    child { node {}}
                  }
                  edge from parent[draw=none]
                }
                edge from parent[draw=none]
                };
                \draw[dashed, line width=2pt, draw opacity=0.3] (-6, -6) -- (6, -6) ;
                \draw[dashed, line width=2pt, draw opacity=0.3] (4, -3) -- (6, -3) ;
                \draw[dashed, line width=2pt, draw opacity=0.3] (-6, -0.8)   -- (6, -0.8);
                \draw[dashed] (0.5, -1.5) -- (6.5, -1.5) ;

                \draw[thick, dotted] (u) -- (v1);
                \draw[thick, dotted] (u) -- (v2);
                \draw[thick, dotted] (u) -- (vt);
                \draw[thick, dotted] (root) -- (u);
                
                \node[labelnode, right=0.3pt of root] {$root$};

                \node[labelnode] at (6.5,-0.8) {$\SL_k$};
                
                \node[labelnode] at (6.5,-3) {$\SL_1$};
                \node[labelnode] at (6.5,-6) {$\SL_0$};

                \draw[decorate, decoration={brace, amplitude=5pt}, thick] (7,-3.1) -- (7,-6) 
                
                node[labelnode, midway, right=6pt, align=center] {Gap of \\ $\gap$ levels};

                \draw[decorate, decoration={brace, amplitude=5pt}, thick] (7,-1.5) -- (7,-2.9)
                node[labelnode, midway, right=6pt, align=center] {Gap of \\ $j <  \gap$ levels};
        \end{tikzpicture}}
        \caption{It shows $\Delta$-ary $\alpha$-HST with special layers marked as $\SL_i$ with gap of $\gap = 2 \log \log \Delta$. Also, for node $u$ immediate special descendants are leaf nodes in $\SL_0$.}
        \label{fig:delta-ary2HST}
\end{figure}
     For any non-leaf node $u$ (that may or may not be in a special layer), we now define the {\em special descendants} of $u$. If height $h(u) < \gap$, the special descendants of $u$ are its proper descendants in $\SL_0$.  If $h(u) \geq \gap$, then let $\SL_i$ denote the special layer with maximum height among special layers at height at most $h(u)$. The special descendants of $u$ are its proper descendants in layers $\SL_0, \SL_1, \ldots, \SL_{i-1}$.

    A crucial property is that for a node $u$ with $h(u) \geq \gap$, and any special descendant $v$ of $u$, the height separation $h(u) - h(v)$ is at least $\gap$.

     For any node $u$ with $h(u) > 0$, the set of {\em immediate special descendants} of $u$, denoted $\isd(u)$, consists of the special descendants of $u$ with maximum height. Note that if $h(u) < 2\gap$, $\isd(u)$ equals the set of special descendants of $u$. If $h(u) \geq \gap$ and $v \in \isd(u)$, then the height separation $h(u) - h(v)$ is at least $\gap$. If $u$ itself is in a special layer $\SL_i$ for $i \geq 1$, then  each $v \in \isd(u)$ is in special layer $\SL_{i-1}$. 
     The size of $\isd(u)$ depends on the difference in height between $u$ and nodes in $\isd(u)$. Thus, if $h(u) <\gap$, $|\isd(u)| \leq \Delta^{j}$ and if $h(u) \geq \gap$, then $|\isd(u)| \leq \Delta^{\gap+j}$; where $0\leq j < \gap$, as shown in Fig. \ref{fig:delta-ary2HST}.

\subsection{The Framework of Bansal et al.}
Bansal et al.~\cite{metriclog2gupta} present a framework that yields an $O(\log n)$-competitive algorithm for online metric matching on a 2-HST. In their framework, they develop an algorithm for metric matching in an intermediate model called the \emph{Restricted Reassignment Online Model.} In the standard online version of the problem, matching decisions are immediate and irrevocable; requests cannot be reassigned. In the more-relaxed restricted-reassignment online model, a request $r$ currently matched to a server
$s$ may be reassigned to another server $s'$, provided some technical condition (``restriction'') is met. For such a reassignment the request $r$ pays a cost of $d(r,s')$. In this model, the cost incurred by the algorithm is the sum of all assignment and reassignment costs of requests. 

Bansal et al.~\cite{metriclog2gupta} show that any algorithm in the restricted reassignment model can be converted into an algorithm in the online model without an increase in cost. Furthermore, they present an $O(\log n)$-competitive algorithm for metric matching in the restricted reassignment online model. We now describe this algorithm, which we refer to throughout as the BBGN algorithm.

\paragraph{BBGN Algorithm.} This algorithm maintains a level $Level(r)$ (resp.~$Level(s)$) for each request $r$ (resp.~server $s$). At the beginning, we set $Level(s) \leftarrow \infty$ for each server $s$, indicating it is unmatched. When a request $r$ arrives, we initially set $Level(r)$ to $0$. For any request $r$ that has arrived, and any $0 \leq \ell \leq \lambda$, the set of available servers for $r$ at height $\ell$ is defined as
\[ A(r, \ell) = \{s \in S \ | \ s \in N(r,\ell) \mbox{ and } Level(s) > \ell \}\]
Here, we use $s \in N(r, \ell)$ to mean that server $s$ is located at a node in $N(r, \ell)$.
   \begin{algorithm}[!htbp]
    \caption{BBGN Algorithm} 
    \label{alg:BBGN}
    \begin{algorithmic}[1]
    \State \textbf{Initialize:} $Level(s) \gets \infty$ for all servers $s$
    \Statex
    \State \textbf{When new request $r$ arrives:}
    \State $Level(r) \gets 0$
    \State \Call{Assign}{$r$}
    \Statex
    \Procedure{Assign}{$r$}
        \State Find minimum height $\ell \geq Level(r)$ such that $A(r,\ell) \neq \emptyset$ \label{line:nextavble}
        \State $s \gets$ a server chosen uniformly at random from $A(r, \ell)$ \label{line:randomchoice}
        \State Match $r$ to $s$; $Level(r) \gets \ell$; $Level(s) \gets \ell$ \label{line:labels}
        \State if $s$ was previously matched to $r'$, then \Call{Assign}{$r'$} \label{line:reassign}
    \EndProcedure
    \Statex
    \end{algorithmic}
    \end{algorithm}

Notice that when we match a request $r$ to a server $s$ that was matched earlier to a request $r'$ (Lines \ref{line:labels}, \ref{line:reassign}), a reassignment of $r'$ is triggered, which is processed by $\ASSIGN(r')$. Thus a newly arriving request can trigger a chain of reassignments. \cite{metriclog2gupta} show this chain terminates because we eventually reach a stage where the server assigned in Line~\ref{line:labels} is unmatched.

We make a few other observations here:
\begin{itemize}
    \item If request $r$ is matched to server $s \in N(r, \ell)$ then at that point $Level(s) = Level(r) = \ell$.
    \item A request's level is non-decreasing over time: when a request $r$ looks to match itself to a server, it picks from $A(r,\ell)$ for $\ell$ that is at least as large as its current $Level(r)$; after picking a server, $Level(r)$ is updated to $\ell$. 
    \item $Level(s)$ for a server $s$ decreases each time it is matched to a new request. This follows because for every $s' \in A(r,\ell)$, $Level(s') > \ell$.
\end{itemize}

Our reassignment algorithm shares the three features of the BBGN Algorithm noted here. We conclude this section by pointing out that for the BBGN Algorithm itself, there is a lower bound of $\Omega(\log n)$ on the competitive ratio on $2$-ary $2$-HSTs. This argument is sketched in Appendix \ref{sec:app-bbgn-lb}.
\section{Our Algorithm for Metric Matching on HSTs}
In this section, we present our algorithm for metric matching on HSTs. We follow the framework of \cite{metriclog2gupta}, in that we first present an algorithm in the reassignment model, and then describe how to convert that into an online algorithm.\footnote{For reasons described in Section~\ref{sec:online-conv}, we do not henceforth refer to  the reassignment model as ``restricted''} Recall that our goal is to explore how that framework can be adapted to get a competitive ratio for $\Delta$-ary $\alpha$-HSTs that is independent of $n$, the number of servers/requests in the instance. Our algorithm for the reassignment model inherits many features of the BBGN algorithm. Before describing it, we overview the key differences from their algorithm, which we summarized above. 

Fix a height $0 \leq \ell \leq \treeheight$ and the timeline of a request $r$ within which it is matched to possibly many servers in $N(r,\ell)$. Let $u$ denote the ancestor node of $r$ at height $\ell$. Our algorithm differs from the BBGN algorithm in two key respects:
\begin{enumerate}
\item When request $r$ searches for a server from the available set $A(r,\ell) \subseteq N(r, \ell)$ {\em for the first time}, how does $r$ pick this server? The BBGN algorithm picks uniformly at random from $A(r, \ell)$, paying a cost of $2d(u,r)$ for the assignment. Our algorithm picks a server from $A(r,\ell)$ using a different distribution, but it still pays a cost of $2 d(u,r)$. As the costs paid are the same, and the first choice from $N(r,\ell)$ happens only once, this difference is not consequential.
\item When request $r$ {\em subsequently} searches for a server from the available set $A(r, \ell)$, how does $r$ pick this server? Note that in this case $r$ was already matched to a server $s \in N(r,\ell)$, but because of a reassignment triggered by some other request, it is looking for a new server $s'$ from $A(r,\ell)$. The BBGN algorithm once again picks uniformly at random from $A(r, \ell)$, and pays a cost of $2 d(u,r)$; see Line~\ref{line:randomchoice} of Algorithm~\ref{alg:BBGN}. Our algorithm instead picks from a distribution on $A(r, \ell)$ that is skewed towards servers that are closer to the server $s$ from which $r$ is being displaced. The cost paid for this reassignment of $r$ is $d(s,s')$, which can be significantly smaller than $2d(r,u)$, and this opens the door to bound the competitive ratio by a factor that is independent of $n$. 
\end{enumerate}

In the remainder of this section, we describe our algorithm in the reassignment model. The cost analysis of the algorithm in the reassignment model is outlined in Section~\ref{sec:overview} and presented in Appendix~\ref{sec:cost}. Appendix~\ref{sec:online-conv} therein explains how to convert this algorithm into an online algorithm without increasing the matching cost.

\subsection{Algorithm in Reassignment Model}  
Like the BBGN algorithm, we maintain $Level(r)$ for each request $r$ and $Level(s)$ for each server $s$ that are initialized as in their algorithm. The set $A(r,\ell)$ of available servers for request $r$ at height $\ell$ is defined exactly as before.

When a new request $r$ arrives, we set $Level(r)$ to $0$ and call the procedure $\ALLOCATE(r)$. This finds the smallest height $\ell \geq Level(r)$ at which $A(r,\ell) \neq \emptyset$. It then picks a server $s$ from $A(r,\ell)$ (Line~\ref{line:findserver}), matches $r$ with $s$, and updates $Level(r)$ and $Level(s)$ (Lines~\ref{line:updatelabel-r} and \ref{line:updatelabels}). If $s$ was previously matched to request $r'$, then we search for a new server for $r'$ using procedure $\REASSIGN(r',s)$. The reader should note the close parallel with the procedure $\ASSIGN(r)$ in Algorithm~\ref{alg:BBGN}. One difference here is that the server $s$ is chosen from $A(r,\ell)$ using a call to $\RSALLOCATION(r,\ell,u)$.

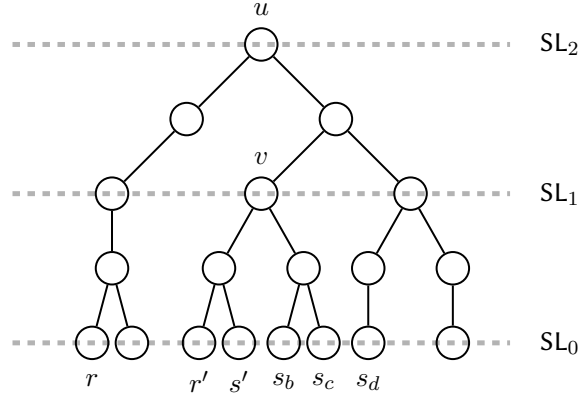
\begin{figure}[!ht]
\centering
\resizebox{0.6\linewidth}{!}{\begin{tikzpicture}[
                every node/.style={circle, thick, draw, minimum size=13pt, inner sep=0pt},
                labelnode/.style={draw=none, circle=none, inner sep=1pt, font=\bfseries\normalfont},
                edge from parent/.style={draw, thick},
                edge from parent path={(\tikzparentnode) -- (\tikzchildnode)}, 
                level 1/.style={sibling distance=3cm},
                level 2/.style={sibling distance=3cm},
                level 3/.style={sibling distance=1.7cm},
                level 4/.style={sibling distance=0.8cm},
                scale=0.7
            ]
            
              \node (root) {}
                child{ node {} 
                    child { node (L1) {}
                        child {  node (L2) {}
                    child { node (r){} }
                    child { node {} }
                  }
                  }
                  child[missing] {}
                }
                child{ node (p) {} 
                child { node (v) {}
                  child { node (x) {}
                    child { node (rp){} }
                    child { node (s){} }
                  }
                  child { node (y) {}
                    child { node (s_b){} }
                    child { node (s_c){} }
                  }
                }
                child { node {}
                  child { node (z) {}
                    child { node (s_d){} }
                  }
                  child { node (w) {}
                    child { node {} }
                  }
                }
                };
                \draw[dashed, line width=2pt, draw opacity=0.3] (-5, -6) -- (5, -6) ;
                \draw[dashed, line width=2pt, draw opacity=0.3] (-5, -3)   -- (5, -3);
                \draw[dashed, line width=2pt, draw opacity=0.3] (-5, 0)   -- (5, 0);
                \node[labelnode] at (6, -6) {$\SL_0$};
                \node[labelnode] at (6, -3) {$\SL_1$};
                \node[labelnode] at (6, 0) {$\SL_2$};

                \node[labelnode, below=1pt of r] {$r$};
                \node[labelnode, below=1pt of rp] {$r'$};
                \node[labelnode, below=1pt of s] {$s'$};
                \node[labelnode, below=1pt of s_b] {$s_b$};
                \node[labelnode, below=1pt of s_c] {$s_c$};
                \node[labelnode, below=1pt of s_d] {$s_d$};
                 \node[labelnode, above=0.5pt of root] {$u$};
                    \node[labelnode, above=0.5pt of v] {$v$};
        \end{tikzpicture}}
        \caption{$r$ was matched to $s'$, and $Level(r)=Level(s')=h(u)=4$. When $r'$ arrives, it is matched to $s'$, triggering reassignment of $r$. We have $A(r,4)=\{s_b, s_c, s_d\}$ at this point. The call to \textsc{Reassign}$(r,s')$ (Line~\ref{line:call-reassign}) will set $v$ in Line~\ref{line:v} as shown; this leads to uniformly picking either $s_b$ or $s_c$ for request $r$.}
        \label{fig:reassign}
\end{figure}

The procedure $\REASSIGN(r,s)$ is called when server $s$ which was matched to $r$ has been assigned to some other request. The purpose of the procedure is to find a new server in $A(r,\ell)$ to assign $r$ to, where $\ell:= Level(r)$. If $A(r,\ell) = \emptyset$ this is clearly not possible; in this case $\ALLOCATE(r)$ is called to find a server for $r$ from $A(r, \ell')$ for some $\ell' > \ell$. If $A(r,\ell) \neq \emptyset$, we look for a server that is close to the previous match $s$. We walk up the path from $s$ to $u$ in the HST to find the first special descendant of $u$ that has descendants in $A(r,\ell)$; if such a node is not found before reaching $u$ we set $v = u$. Now $v$ is guaranteed to have descendants in $A(r,\ell)$. A server $s'$ from such descendants is found by a call to $\RSALLOCATION(r,\ell,v)$. See Fig.~\ref{fig:reassign}. Request $r$ is matched to $s'$, and $Level(s')$ is updated to $Level(r)$. If $s'$ was previously matched to some request $r'$, then we call $\REASSIGN(r',s')$ now.

 \begin{algorithm}[H]
    \caption{New Reassignment Algorithm} 
    \label{alg:NewStrategy}
    \begin{algorithmic}[1]
    \State \textbf{Initialize:} $Level(s) \gets \infty$ for all servers $s$
    \Statex
    \State \textbf{When new request $r$ arrives:}
    \State $Level(r) \gets 0$
    \State \Call{Allocate}{$r$}
    \Statex
    \Procedure{Allocate}{$r$}
        \State Find minimum height $\ell \geq Level(r)$ such that $A(r,\ell) \neq \emptyset$
        \State $Level(r) \gets \ell$ \label{line:updatelabel-r}
        \State $u \gets \mbox{ ancestor of $r$ at height $\ell$} $
        \State $s \gets$ \Call{RSAllocation}{$r, \ell, u$} \label{line:findserver}
        \State Match $r$ to $s$; $Level(s) \gets \ell$  \label{line:updatelabels} // assignment cost $= d(r,s) = 2d(u,r)$
        \State if $s$ was previously matched to some $r'$, then \Call{Reassign}{$r', s$} \label{line:call-reassign}
    \EndProcedure
    \Statex
    
    \Procedure{RSAllocation}{$r, \ell, v$} \label{line:rs_fn}
        \State $w \gets v$
        \While{$w \not\in \SL_0$}
            \State \label{line:randomchoice-ours} Pick a node uniformly at random from \begin{center} $\{x \in \isd(w) \ | \ \mbox{some server in } A(r, \ell) \mbox{ is a descendant of } x \}$ \end{center} 
            \State Update $w$ with chosen node. \label{line:updateW} \label{16}
        \EndWhile
        \State $s \gets$ any server from $A(r,\ell)$ at leaf $w$
        \State Return $s$
    \EndProcedure
    \Statex
    \Procedure{Reassign}{$r, s$} \label{line:reassign_fn}
        \State $u \gets$ ancestor of $r$ at height $\ell: = Level(r)$
        \If{$A(r,\ell) = \emptyset$}
            \State \Call{Allocate}{$r$}
        \Else
          \State $B \gets \{u \}\  \cup$ Special descendants of $u$ that are ancestors of $s$
          \State $v \gets $ smallest height node in $B$ whose descendants contain a server from $A(r,\ell)$ \label{line:v}
          \State $s' \gets$ \Call{RSAllocation}{$r, \ell, v$}
          \State Match $r$ to $s'$; $Level(s') \gets \ell$ \label{line:reassign-ours} // Reassignment cost $= d(s,s') \leq 2 d(v,s)$
          \State if $s'$ was previously matched to some $r'$, then \Call{Reassign}{$r', s'$}
        \EndIf
    \EndProcedure
    \end{algorithmic}
    \end{algorithm}

Whenever a call to procedure $\RSALLOCATION(r, \ell, v)$ is made, the precondition that $v$ has descendants from $A(r,\ell)$ holds. A simple randomized descent finds and returns a server in $A(r,\ell)$ that is a descendant of $v$. This descent is like the Randomized Subtree algorithm of \cite{couplingLineGupta12}, except that it only uses special descendants.

\paragraph{Important Note on Cost.} When request $r$ gets assigned a server via $\ALLOCATE(r)$, then for some height $\ell$, it gets assigned (Line~\ref{line:updatelabels}) a server $s$ from $A(r,\ell)$ for the first time. This assignment costs $d(r,s)$. (If $r$ was matched to a server $s''$ prior to this, then we have (a) $s'' \in T(r, \ell -1)$, and (b) $s \in T(r, \ell) \setminus T(r, \ell -1)$. Thus, $d(s'',s) = d(r,s)$, and the cost we charge is at least the distance between $s''$ and $s$.) When request $r$ gets assigned a server $s'$ in
$\REASSIGN(r,s)$ (Line~\ref{line:reassign-ours}), this reassignment of $r$ from $s$ to $s'$ costs the algorithm $d(s,s') \leq 2 d(v,s)$. 

The following observations made for the BBGN algorithm are readily seen to hold here as well: (a) When request $r$ is matched to server $s \in N(r,\ell)$, $Level(s) = Level(r) = \ell$; (b) For any request $r$, $Level(r)$ is non-decreasing with time; (c) $Level(s)$ for a server $s$ decreases each time it is matched to a new request. For completeness, we argue that the sequence of reassignments triggered by the arrival of a new request terminates. The main point is that when a request $r$ gets assigned a server $s$ that was previously assigned to request $r'$, then at that time $Level(r') > Level(r)$. As the level of a request only changes in its search for a server, this means that no subsequent request in the chain can trigger a reassignment of $r$. This implies termination.

We present one final observation that we will need later.
\begin{lemma}\label{lem:previous-levels}
 Suppose that at some point during the algorithm, we have $Level(r) = \ell$ for some request $r$. Then for any smaller height $0 \leq \ell' < \ell$, we have $A(r,\ell') = \emptyset$ at that time. 
\end{lemma}

\begin{proof}
    When $Level(r)$ became greater than $\ell'$ for the first time, it was the case that $A(r, \ell') = \emptyset$. As server levels do not increase with time, $A(r, \ell')$ remains empty subsequently.
\hfill \qed \end{proof}

The above reassignment algorithm can be converted into an online algorithm without an increase in cost. This is similar to \cite{metriclog2gupta}, but with a subtlety owing to a less restrictive reassignment model; see Appendix~\ref{sec:online-conv}.
\section{Overview of Analysis}
\label{sec:overview}
In this short section, we give a high-level overview of how we upper bound the expected cost incurred by the reassignment algorithm (Algorithm~\ref{alg:NewStrategy}) by $O((\log \Delta) \log \log \Delta)$ times the cost of the optimal (offline) matching on any instance. Just as in \cite{metriclog2gupta}, the final matching computed by our reassignment algorithm is optimal. Thus, it suffices to consider each request $r$ separately, and show that the expected total cost of the initial assignment and subsequent reassignments undergone by request $r$, is at most $O((\log \Delta) \log \log \Delta)$ times the distance of $r$ to its assigned server in the final optimal matching. 

Suppose $z$ is a proper ancestor of request $r$ such that at the end of our reassignment algorithm, $Level(r)$ equals $h(z)$, the height of $z$ in the tree. Then the distance of $r$ to its assigned server in the final matching is $O(\alpha^{h(z)})$, where we recall that $\alpha^{h(z) - 1}$ is the distance from $z$ to its child.
Suppose that $u$ is an arbitrary proper ancestor of $r$ in the HST, and at some point during Algorithm~\ref{alg:NewStrategy}, $Level(r)$ becomes equal to $h(u)$. During the period when $Level(r)$ equals $h(u)$, $r$ may get assigned to multiple servers in the subtree at $u$; it suffices to show that the sum of the (assignment and reassignment) costs charged for all such assignments is at most $O((\log \Delta) \log \log \Delta) \cdot \alpha^{h(u)}$ (Lemma~\ref{lemma:requestcost}).

Let us focus on the reassignment costs incurred by request $r$ when $Level(r)$ equals $h(u)$. As already noted, our algorithm tries to reassign more ``locally'' than the algorithm in \cite{metriclog2gupta}. This however has the effect of increasing the number of reassignments, compared to \cite{metriclog2gupta}, for request $r$ when $Level(r) = h(u)$. In particular, there are two types of reassignments: (a) a reassignment from a server in a subtree of some immediate special descendant of $u$ to a server in a subtree of a {\em different} immediate special descendant of $u$; and (b) for some special descendant $v$ of $u$, a reassignment from a server in a subtree of some immediate special descendant of $v$ to a server in a subtree of a {\em different} immediate special descendant of $v$. Note that the cost of reassignments of type (b) depends on $h(v)$; as we make the height $h$ smaller, there are more nodes $v$ to consider, but the corresponding reassignments are cheaper. However, provided we get the ``right'' bound on the total cost of reassignments of type (a), and for type (b) for each special descendant $v$ of $u$, we can show that these costs add up in a way that gives us our bound of $O((\log \Delta) \log \log \Delta) \cdot \alpha^{h(u)}$. This addition of costs is similar to one done by Gupta and Lewi \cite{couplingLineGupta12} in their analysis of the randomized subtree algorithm, and the calculation in Lemma~\ref{lemma:requestcost} follows their lead.

To complete the analysis, we focus on getting the ``right'' bound on the cost of all reassignments of type (a). Reassignments of type (b) for each special descendant $v$ of $u$ are handled analogously.
Recall that the size of $\isd(u)$, the set of immediate special descendants of $u$, is at most $\Delta^{2 \gap}$, where $\gap = 2 \log \log \Delta$. To get the ``right'' bound as above, we need to show that the number of nodes $w$ in $\isd(u)$, such that request $r$ is matched to servers in the subtree at $w$, is only $O( \log \Delta^{2 \gap} )$ (as stated in Corollary~\ref{clm:logbound}). Suppose that $k \leq \Delta^{2 \gap}$ of the nodes in $\isd(u)$ have available servers for $r$ at some point when $r$ is about to make a uniform random choice of which of these $k$ nodes to descend into to get assigned a server (in Line~\ref{line:randomchoice-ours} of Algorithm~\ref{alg:NewStrategy}). It suffices to argue that by the time $r$ gets reassigned out of the node $x$  it randomly chose to descend into, $\Omega(k)$ of the $k$ nodes, in expectation,  become unavailable for $r$ (Lemma~\ref{lem:numberofjumpsF}). The difficulty with proving this claim is that the order in which these $k$ nodes become unavailable is a random variable influenced not only by the random choices made by other requests but also by $r$. Furthermore, the argument has to be robust enough to handle some complexity introduced by the way we defined special descendants, a choice we made to optimize the dependence on $\Delta$.

We handle these issues by considering an arbitrary pair $x,y$ of these $k$ nodes, and using a coupling-like argument to show that no matter what the random choices made by the other requests are, we have a significant probability that if $r$ descends into one of these two nodes, the other node will have become unavailable by the time $r$ is reassigned out of the chosen node (Lemma~\ref{lem:key}). Apriori, we do not know which of the nodes $x$ and $y$ the ``good'' node is, but the coupling-like argument establishes that at least one of them has to be ``good''. This is the technically most involved and novel part of the paper. 
The full formal analysis, of which this section offered an overview, is given in Appendix~\ref{sec:cost}.

\newpage

\bibliographystyle{splncs04}
\bibliography{references.bib}

\newpage
\appendix
\section*{Appendix}
\section{Analysis of the Reassignment Algorithm} \label{sec:cost}
In this section, we upper bound the cost incurred by the reassignment algorithm (Algorithm~\ref{alg:NewStrategy}) in terms of the cost of the optimal (offline) matching on any instance. 

We begin our analysis by first presenting two observations from \cite{metriclog2gupta} about the costs of a reassignment algorithm.

\subsection{Local Allocation Property} \label{SS:4.1}

Consider any node $v$ in the HST $T$ and its subtree $T_v$. 
\begin{itemize}
    \item Let $\cR(v)$ be the set of requests that arrived in subtree $T_v$.  
    \item Let $\cS(v)$ be the set of servers located in the subtree $T_v$.
\end{itemize} 

Define the excess 
\[E(v)=\max\{0,|\cR(v)|-|S(v)|\}\]
This is the number of requests from $\cR(v)$ that have to be matched to a server outside $T_v$ (by any matching). Recall that the length of the edge connecting $v$ to its parent is $\alpha^{h(v)}$, where $h(v)$ is the height of $v$. The following lower bound on the cost of the optimal offline matching has been used in several previous works.

\begin{lemma}[Lower Bound] \label{lem:lower-bound}
   The cost of any matching for the input instance is at least 
   \[\sum_{v\in T} E(v)\cdot 2\alpha^{h(v)}\]
\end{lemma}

\begin{proof}
    Fix a matching of the requests to servers. Consider any request $r$, and suppose it is matched to server $s$. Suppose $u$ is the least common ancestor of $r$ and $s$ in $T$. To every edge on the path from $r$ to $u$, let $r$ charge twice the length of the edge. The total amount charged by request $r$ is $d(r,s) = 2d(r,u)$. Thus the total amount charged by all the requests is the cost of the matching.

    Now consider the edge from some node $v$ in the HST to its parent. At least $E(v)$ requests charge this edge, and each such request charges this edge $2 \alpha^{h(v)}$. This completes the proof. 
\hfill \qed \end{proof}

We now show that the final matching produced by Algorithm~\ref{alg:NewStrategy} enjoys a strong locality property. This property holds for our algorithm because it inherits the basic structure of the BBGN algorithm. 

\begin{lemma}[Local Allocation]\label{lem:local_requests}
Consider the final matching produced by Algorithm~\ref{alg:NewStrategy} and any node $v$ in the HST. Exactly $E(v)$ requests from $\cR(v)$ are matched to a server outside $T_v$ (i.e., matched to a server not in $\cS(v)$).
\end{lemma}

\begin{proof}
Consider any point during the algorithm execution where some server $s \in \cS(v)$ is either unmatched or matched to a request that is not in $\cR(v)$. Now consider any request $r \in \cR(v)$ that has already arrived, and let $w$ be the least common ancestor of $r$ and $s$ in the HST $T$. As $s \in A(r, h(w))$, we have $A(r, h(w)) \neq \emptyset$. From Lemma~\ref{lem:previous-levels}, we have $Level(r) \leq h(w) \leq h(v)$. So $r$ is matched to a server in $\cS(v)$. That is, every request $r \in \cR(v)$ that has arrived by now is matched to a server in $\cS(v)$. As we'll see, the Lemma follows from this observation.

Consider the case that $|\cR(v)| < |\cS(v)|$. At the end of the algorithm, there is at least one server $s \in \cS(v)$ that is either unmatched or matched to a request not in $\cR(v)$. Thus, by the above observation, every request in $\cR(v)$ is matched to a server in $\cS(v)$, and the Lemma holds.

Now consider the case that $|\cR(v)| \geq |\cS(v)|$. Consider the matching at the end, after all requests have arrived. At this point in time, there cannot be a server $s \in \cS(v)$ that is either unmatched or matched to a request not in $\cR(v)$. Otherwise, the above observation implies that all $|\cR(v)|$ requests from $\cR(v)$ are matched to servers in $\cS(v) \setminus \{s \}$, which is impossible as $|\cR(v)| \geq |\cS(v)|$.

Thus, at the end, all servers in $\cS(v)$ are matched to requests in $\cR(v)$. Thus there are exactly $E(v) = |\cR(v)| - |\cS(v)|$ requests from $\cR(v)$ that are matched to servers not in $\cS(v)$.
\hfill \qed \end{proof}

\subsubsection{Cost of Algorithm in Reassignment Framework.}
\def\reqcost{\textsf{ReqCost}}

Let $u$ be some node in the HST. We say that a request $r \in \cR(u)$ {\em reaches}
$u$ if $Level(r)$ attains a value of at least $h(u)$ when the algorithm terminates. Let $\zeta(u) \subseteq \cR(u)$ denote the set of requests that reach $u$. Thus, if $u$ is an internal node, it follows from Lemma~\ref{lem:local_requests} that
\begin{equation}
 |\zeta(u)| = \sum_{v \in \mathrm{children}(u)} E(v) \label{eq:zeta}  
\end{equation}
Let $u$ be some node in the HST, and $r \in \cR(u)$ a request. We define $\reqcost(r,u)$ to be the total cost incurred by the algorithm in assigning servers from $A(r, h(u))$ to $r$. If $r$ does not reach $u$, $\reqcost(r,u) = 0$. For an $r$ that reaches $u$, recall that the first time $r$ is assigned a server from $A(r,h(u))$, that costs $2d(r,u)$. Subsequently, when $r$ is reassigned from server $s \in A(r,h(u))$ to another server $s' \in A(r,h(u))$, there is a cost of $d(s,s')$. The quantity $\reqcost(r,u)$ is the sum of the cost of the first assignment to a server in $A(r,h(u))$, and the costs of subsequent reassignments to servers in $A(r,h(u))$. (Note that for $r \in \zeta(u)$, it is still possible that $\reqcost(r,u) = 0$.)

In Appendix~\ref{SS:4.2}, which contains the technical core of the paper, we show that for some constant $c > 0$,
\begin{equation}
    \E[\reqcost(r,u) \mid r \in \zeta(u)] \leq c \cdot \log \log \Delta \cdot \log \Delta \cdot \alpha^{h(u) - 1} \label{eq:reqcost}
\end{equation}
In the remainder of this section, we use this to bound the competitive ratio of our reassignment algorithm. 

The total cost incurred by Algorithm~\ref{alg:NewStrategy} is
\[ \sum_{u \in T \setminus\{leaves\}} \sum_{r \in \zeta(u)} \reqcost(r,u)\]
The expected total cost is therefore
\[ \sum_{u \in T \setminus\{leaves\}} \sum_{r \in \cR(u)} \E[\reqcost(r,u) | r \in \zeta(u)] \cdot \Pr[r \in \zeta(u)] \]
Using~Eq.~\ref{eq:reqcost}, this is is upper bounded by
\begin{eqnarray*}
 & & O((\log \log \Delta) \log \Delta) \sum_{u \in T \setminus\{leaves\}} \alpha^{h(u) - 1} \sum_{r \in \cR(u)} \Pr[r \in \zeta(u)] \\
 & \leq &  O((\log \log \Delta) \log \Delta) \sum_{u \in T \setminus\{leaves\}} \alpha^{h(u) - 1} |\zeta(u)| \\
 & \leq & O((\log \log \Delta) \log \Delta) \sum_{u \in T \setminus\{leaves\}} \alpha^{h(u) - 1} \sum_{v \in \mathrm{children}(u)} E(v) \\
 & \leq & O((\log \log \Delta) \log \Delta) \sum_{v \in T} \alpha^{h(v)} E(v) \\
 & \leq & O((\log \log \Delta) \log \Delta) \cdot \OPT
\end{eqnarray*}

Here, the second inequality comes from~Eq.~\ref{eq:zeta} and the last one from Lemma~\ref{lem:lower-bound}.

\begin{theorem}\label{thm:comp-in-reassign}
    The expected total cost of the assignments and reassignments made by Algorithm~\ref{alg:NewStrategy} is $O((\log\log \Delta) \cdot  \log \Delta) \cdot \OPT$ on any $\Delta$-ary $2$-HST. Here, $\OPT$ is the cost of the optimal offline matching of the instance on the HST $T$. Furthermore, recall that Algorithm~\ref{alg:NewStrategy} is an algorithm in the reassignment model. 
\end{theorem}

\subsection{Bounding $\reqcost(r,u)$} \label{SS:4.2}
\def\SpServU{\Gamma}
We introduce some terminology for the analysis here. Let $r'$ be any request, and $x$ be a node in some special layer.
\begin{itemize}
    \item We say that $r'$ {\em descends} into $x$ if within some invocation of procedure $\RSALLOCATION(r',\cdot,\cdot)$, $x$ is chosen as a consequence of the uniform random choice in Line~\ref{line:randomchoice-ours} and assigned to $w$ in Line~\ref{line:updateW}. This means that $r'$ is about to be assigned 
to a server from $T_x$. Let $u$ denote the least common ancestor of $r'$ and $x$; at this time, then, $Level(r') = h(u)$, and $x$ is a special descendant of $u$. We use $r' \downarrow x$ as notation for ``$r'$ descends into $x$''.
    \item Once $r'$ is assigned a server from $T_x$ it continues to be matched to a server in $A(r',h(u)) \cap \cS(x)$ as long as $A(r', h(u)) \cap \cS(x) \neq \emptyset$. This is an important notion of locality that our algorithm enjoys. Once $A(r', h(u)) \cap \cS(x)$ has become empty, any reassignment of $r'$ will then cause it to be assigned to a server not in $T_x$. We say that $r'$ {\em ascends out of} $x$ (denoted by $r' \uparrow x$) at this time. 
\end{itemize}
  Thus, $r'$ descends into $x$ at most once; if it descends into $x$ then it ascends out of $x$ at most once. In between these two events $r'$ is assigned only to servers within the subtree $T_x$.

  For a request $r'$, and a node $w$ such that $r' \not\in \cR(w)$, we say that $w$ is {\em eligible for} $r'$ if there is at least one server $s \in \cS(w)$ with $Level(s) > Level(r')$. 

Let $r$ be an arbitrary request, and $u$ be a non-leaf ancestor node of $r$ such that $r$ reaches $u$, i.e., $Level(r) \geq h(u)$ eventually. 
In this section we bound the expected value of $\reqcost(r,u)$. It can happen that $Level(r)$ is never set to $h(u)$ for such a request $r$; this transpires if $A(r,h(u))$ is empty when $Level(r)$ first becomes greater than $h(u) - 1$. If $Level(r)$ never equals $h(u)$, then $\reqcost(r,u) = 0$. Thus, we will assume that $Level(r)$ equals $h(u)$ at some point; recall that $\reqcost(r,u)$ is the sum of the costs of assigning servers to $r$ when $Level(r) = h(u)$.

Let $\SpServU$ denote the set consisting of $u$ and every special descendant of $u$ that {\em is not} in $T(r, h(u) -1)$. Note that the set of leaves in $\SpServU$ is exactly $N(r,h(u))$. 

\subsubsection{Analysis for Fixed Special Descendant of $u$.}
\label{sec:fixed-v}
Consider any node \(v \in \SpServU\) such that $r$ descends into some node in $\isd(v)$. To bound \reqcost\((r,u)\), we would like to bound the expected number of immediate special descendants of $v$ (that is, nodes in $\isd(v)$) into which $r$ descends; see Corollary~\ref{clm:logbound} below. For this, it suffices to show that every time 
\(r\) descends into a node in $\isd(v)$, by the time it ascends out of it, the number of eligible nodes (for $r$) in $\isd(v)$ is reduced by more than half.

Fix a particular instant where $r$ is about to choose a node from node $\isd(v)$ to descend into. We define 
the following three events:

\begin{itemize}
    \item \({\sf E}_{\downarrow w}:=~r\text{ descends into }w\).
    \item \({\sf E}_{\downarrow w \uparrow}:=\)  \({\sf E}_{\downarrow w}\) happens and then, in some future time instant, $r$ or any other \(h(u)\)-level request ascends out of $w$ for the first time.
    \item \({\sf E}_{\downarrow \uparrow}:=\) $r$ descends into some node and later some  $h(u)$-level request ascends out of the same node for the first time.
\end{itemize}

In the remainder of Appendix~\ref{SS:4.2}, whenever we say that a node $w$ is eligible (resp. ineligible), without mentioning a corresponding request, we mean that it is eligible (resp. ineligible) for $r$. Thus, a node $w \in \isd(v)$ is eligible iff there is at least one server $s \in \cS(w)$ such that $Level(s) > h(u)$. 

For any two  \(x,y \in \isd(v)\), we define the following event,
which is an event in the probability space ${\sf E}_{\downarrow x}$:
\begin{equation}
    \label{eqn:keyQuantity}
    {\sf Q}_{x,y}:= {\text{ Either }} {\sf E}_{\downarrow x \uparrow} {\text{ does not happen, or }} 
y {\text{ is ineligible}} \text{ after } {\sf E}_{\downarrow x \uparrow} {\text{ happens}}. 
\end{equation}

The following technical lemma is at the heart of our proof. 
\begin{restatable}{lemma}{mainlemma}
\label{lem:key}
    For any \(x,y\), we have:
    \[ \Pr\big[{\sf Q}_{x,y}\big] + \Pr\big[{\sf Q}_{y,x}\big] \geq 1\]
\end{restatable}

Let \({\sf N}_{\text{eligible}}\) be the number of eligible nodes in $\isd(v)$. Note, 
\begin{lemma}\label{lem:numberofjumpsF}
\[
\E\big[ 
{\sf N}_{\text{eligible}} \text{ at end of } {\sf E}_{\downarrow \uparrow}
\big] 
~~\leq~~
\frac{1}{2}\cdot  
{\sf N}_{\text{eligible}} \text{ at start of } {\sf E}_{\downarrow \uparrow}
\]
\end{lemma}
\begin{proof}
Fix any event \({\sf E}_{\downarrow \uparrow}\). 
Let \(S\) be the set of 
eligible nodes in $\isd(v)$ that the request $r$ could descend into, just before the beginning 
\({\sf E}_{\downarrow \uparrow}\). Clearly, we have \( |S|={\sf N}_{\text{eligible}} \text{ at start of } {\sf E}_{\downarrow \uparrow}\). We have,
  \begin{align}\label{updown1}
      \E\big[ 
{\sf N}_{\text{eligible}} \text{ at end of } {\sf E}_{\downarrow \uparrow}
\big] 
      &=\sum_{x\in S} \Pr[{\sf E}_{\downarrow x }]\cdot 
      \E\big[ 
{\sf N}_{\text{eligible}} \text{ after } {\sf E}_{\downarrow x \uparrow}~|~{\sf E}_{\downarrow x }
\big] \notag
\\&=\frac{1}{|S|}\cdot \sum_{x\in S} 
      \E\big[ 
{\sf N}_{\text{eligible}} \text{ after } {\sf E}_{\downarrow x \uparrow}~|~{\sf E}_{\downarrow x }
\big]\notag
\\&=\frac{1}{|S|}\cdot 
\sum_{x\in S}
\sum_{y\in S\setminus\{x\}}
 \Pr\big[ 
y {\text{ is eligible}} \text{ after } {\sf E}_{\downarrow x \uparrow}~|~{\sf E}_{\downarrow x}
\big]
  \end{align}  
We have the first equality because any \({\sf E}_{\downarrow \uparrow}\) must happen due to 
descent of \(r\) to some eligible child \(x\) and then subsequently, first time ascension of some \(h(u)\)-level request out of \(x\). 

Let ${\sf Q}^c_{x,y}$ denote the complement of ${\sf Q}_{x,y}$ in the probability space ${\sf E}_{\downarrow x}$. By using the Definition.~\ref{eqn:keyQuantity}, the Eq.~\ref{updown1} further equals to,
\begin{align*}
    Eq.~\ref{updown1} &= \frac{1}{|S|} \cdot \sum_{\{x,y\} \in \binom{S}{2}} 
    \Pr [{\sf Q}^c_{x,y}] + \Pr[{\sf Q}^c_{y,x}]
    \\&\leq  \frac{1}{|S|} \cdot \sum_{\{x,y\} \in \binom{S}{2}} 1 && \text{ By Lemma~\ref{lem:key}}. 
    \\&= \frac{1}{|S|}  \cdot {|S| \choose 2}
    \\&= \frac{|S|-1}{2} 
\end{align*} 
The claim follows. \hfill \qed
\end{proof}
Now we prove our key claim. 
\mainlemma*
\begin{proof}
First, using the complementary event of ${\sf Q}_{y,x}$ in the probability space ${\sf E}_{\downarrow y}$, we can rewrite the key quantity on the left side as:
\[ \Pr\big[{\sf Q}_{x,y}\big] + \Pr\big[{\sf Q}_{y,x}\big]= 1+ 
\big(\Pr\big[{\sf Q}_{x,y}\big] -\Pr\big[{\sf Q}^c_{y,x}\big]\big)
\]
Therefore it suffices to show that 
\begin{equation}\label{eq:compare}
\Pr\big[{\sf Q}_{x,y}\big] \geq \Pr\big[{\sf Q}^c_{y,x}\big] 
\end{equation}

We show this inequality for each history of the algorithm that leads to the choice facing $r$ at the moment: whether to descend into $x$ or $y$ or some other immediate special descendant ($\isd$) of $v$. We will need to map subsequent algorithmic traces within ${\sf Q}^c_{y,x}$ to traces within ${\sf Q}_{x,y}$. For this purpose, we make certain simplifications. We assume that no past or future request arrives at a leaf that is not within subtree $T_u$. This is without loss of generality because any such request that is matched to a server in $\cS(u)$ will have level greater than $h(u)$, and thus will not affect the events under consideration here. We can delete past requests that arrived outside of subtree $T_u$ and set the levels of the servers matched to them to $\infty$.

We are interested in the algorithmic traces in two scenarios: when the request $r$ under consideration descends into $x$, and when $r$ descends into $y$. One further simplification we need to make is that for requests of level $h(u)$ that are assigned to either a server in $\cS(x)$ or $\cS(y)$, we will not track the precise server assigned to those requests. We will only record that such requests are assigned to either $x$ or $y$, as appropriate. 
For any server $s$ in $\cS(x)$ or $\cS(y)$ that the algorithm has in the past assigned to a request at level $h(u)$, we set $Level(s)$ to $\infty$. Thus,

\begin{enumerate}
    \item The node $x$ is eligible for $h(u)$-level requests if the number of requests at level $h(u)$ assigned to $x$ is smaller than the number of servers in $\cS(x)$ with level $\infty$.
    \item  request with level $h(u)$ ascends out of $x$ for the first time when a request with level less than $h(u)$ is assigned to a server in $\cS(x)$ with level $\infty$, {\em after} $x$ has become ineligible for $h(u)$-level requests.
\end{enumerate}

The same notions are defined similarly for node $y$. We will stop the algorithm when a request with level $h(u)$ ascends out of $x$, or $y$, for the first time. Stopping the algorithm at these instants is enough to determine if the relevant event ${\sf Q}_{x,y}$ or ${\sf Q}_{y,x}$ happens.

We can now formalize an algorithmic trace using the notion of a {\em transcript}. A transcript records the full sequence of the following two types of events: (i) an arrival event, when a new request arrives at some leaf; (ii) a random choice event, where some request $r'$  makes a uniform random choice as to which eligible $\isd$ of some node $w'$ to descend into. In the world where $r$ is assigned to $x$ (resp. $y$), any transcript begins with the next arrival event. Transcripts may be partial, that is, they need not go on till the algorithm stops. But a transcript contains the full history of past random choices.

Now consider a transcript $\sigma$ and a random choice event in it, where some request $r'$ makes a uniform random choice as to which eligible $\isd$ of some node $w'$ to descend into. Supposing there are $k$ such eligible $\isd$s, the probability that $r'$ made the particular choice that is described in the random choice event is $1/k$. Note that $\Pr[\sigma]$ is the product of the probabilities of the random choice events in $\sigma$. 

Let $\Sigma$ be the set of all transcripts $\sigma$ such that in the world where $r \downarrow y$, there is at least one extension of $\sigma$ that results in the event ${\sf Q}_{y,x}^c$. That is, $\Pr \big[ {\sf Q}_{y,x}^c ~|~r \downarrow y, \sigma \big] > 0$. We show by induction that for each $\sigma \in \Sigma$,
\begin{equation} \label{eq:compare-sigma}
 \Pr \big[ {\sf Q}_{x,y} ~|~r \downarrow x, \sigma \big]   \geq \Pr \big[ {\sf Q}_{y,x}^c ~|~r \downarrow y, \sigma \big]
\end{equation}

Notice that for the case where $\sigma$ is the empty transcript, Inequality~\ref{eq:compare-sigma} yields Inequality~\ref{eq:compare}, thus completing the proof.

The following observation is used repeatedly below.

\begin{lemma}\label{lemma:x-available}
  For any $\sigma \in \Sigma$, (a) $x$ is still eligible for $h(u)$-level requests after $\{ r \downarrow y, \sigma \}$; (b) $\sigma$ is a valid transcript in the world where $r \downarrow x$, that is, $\Pr \big[ \sigma ~|~r \downarrow x \big] > 0$.
\end{lemma}

\begin{proof}
    Part (a) follows because $\Pr \big[ {\sf Q}_{y,x}^c ~|~r \downarrow y, \sigma \big] > 0$. For Part (b), the key point is that $r \downarrow x$ only affects the eligibility of $x$ for $h(u)$-level requests. Furthermore, as $x$ is still eligible for $h(u)$-level requests after $\{ r \downarrow y, \sigma \}$, $\{ r \downarrow x, \sigma \}$ never triggers the ascent out of $x$ of any $h(u)$-level request.
\hfill \qed \end{proof}

The base case for the inductive proof of Inequality~\ref{eq:compare-sigma} is when $\sigma$
is maximal in $\Sigma$: there is no proper extension of $\sigma$ that is also in $\Sigma$. 
This means that $\{ r \downarrow y, \sigma \}$ causes the ascent out of $y$ of some $h(u)$-level request. This means that,

\begin{enumerate}

\item The last event in $\sigma$ is either (a) the arrival of a new request at a leaf resulting in the assignment of that request to a server in $\cS(y)$ with level $\infty$ at that same leaf, or (b) a random choice event where some request at level strictly less than $h(u)$ descends to a leaf, and is assigned a server in $\cS(y)$ with level $\infty$ at that same leaf. 

\item Prior to this last event, $y$ had already become ineligible for requests at level $h(u)$.
\end{enumerate}

Now consider the transcript $\sigma $ in the world where $r \downarrow x$. After $\{r \downarrow x, \sigma \}$, $y$ becomes ineligible for $h(u)$-level requests as a consequence of the last event in $\sigma$. Thus, $\Pr \big[ {\sf Q}_{x,y} ~|~r \downarrow x, \sigma \big] = 1$, and Inequality~\ref{eq:compare-sigma} holds. 

Now consider the inductive step, where $\sigma$ is not maximal in $\Sigma$. Here the proof splits into two cases, based on whether the next event is an arrival event or a random choice event.

\begin{enumerate}
    \item After $\{ r \downarrow y, \sigma \}$, the algorithm waits for an arrival event. Let us denote by $e$ the event where the next request arrives. Then by Lemma~\ref{lemma:x-available}(b), after $\{ r \downarrow x, \sigma \}$, the next event is also the arrival event $e$. 

    Then,
    \begin{eqnarray*}
        \Pr \big[ {\sf Q}_{x,y} ~|~r \downarrow x, \sigma \big]   & = &  \Pr \big[ {\sf Q}_{x,y} ~|~r \downarrow x, \sigma e  \big] \cdot \Pr \big[ \sigma e ~|~r \downarrow x, \sigma \big] \\
        & \geq &  \Pr \big[ {\sf Q}_{y,x}^c  ~|~r \downarrow y, \sigma e  \big] \cdot \Pr \big[ \sigma e ~|~r \downarrow x, \sigma \big] \\
        & = &  \Pr \big[ {\sf Q}_{y,x}^c  ~|~r \downarrow y, \sigma e  \big] \cdot \Pr \big[ \sigma e ~|~r \downarrow y, \sigma \big] \\
        & = &  \Pr \big[ {\sf Q}_{y,x}^c ~|~r \downarrow y, \sigma \big]
    \end{eqnarray*}

    Here, the inequality follows from the inductive hypothesis, and the next equality because $ \Pr \big[ \sigma e ~|~r \downarrow x, \sigma \big] = \Pr \big[ \sigma e ~|~r \downarrow y, \sigma \big] = 1$. 

    \item After $\{ r \downarrow y, \sigma \}$, the algorithm is about to choose, for some request $r'$, and for some node $w'$, which eligible node in $\isd(w')$ to descend into. Let $D \subseteq \isd(w')$ denote the subset of eligible nodes for $r'$ at this point, and $\hat{D} = \{ z \in D~|~ \sigma; r' \downarrow z \in \Sigma \}$. Note that by Lemma~\ref{lemma:x-available}, $\Pr \big[ \sigma; r' \downarrow z~|~ r \downarrow x \big] > 0$ for each $z \in \hat{D}$. Let $F \subseteq \isd(w')$ denote the subset of eligible nodes for $r'$ after $\{ r \downarrow x, \sigma \}$. So we have $\hat{D} \subseteq F$. 

    Note that if $Level(r') < h(u)$, or if $Level(r') = h(u)$ and $w' \neq v$, then $D = F$. Recall that $v \in \SpServU$ is the fixed node we are analyzing, and $x, y \in \isd(v)$. Our proof splits into two cases here:
    \begin{enumerate}
        \item $|D| \geq |F|:$ Here, we have \begin{eqnarray*}
            \Pr \big[ {\sf Q}_{x,y} ~|~r \downarrow x, \sigma \big] & \geq & \sum_{z \in \hat{D} } \Pr \big[ {\sf Q}_{x,y} ~|~r \downarrow x, \sigma, r' \downarrow z \big]\cdot \Pr \big[r' \downarrow z  ~|~r \downarrow x, \sigma  \big] \\
            & \geq & \sum_{z \in \hat{D} } \Pr \big[ {\sf Q}_{y,x}^c ~|~r \downarrow y, \sigma, r' \downarrow z \big]\cdot \Pr \big[r' \downarrow z  ~|~r \downarrow x, \sigma  \big] \\
            & \geq & \sum_{z \in \hat{D} } \Pr \big[ {\sf Q}_{y,x}^c ~|~r \downarrow y, \sigma, r' \downarrow z \big]\cdot \Pr \big[r' \downarrow z  ~|~r \downarrow y, \sigma  \big] \\
            & = & \Pr \big[ {\sf Q}_{y,x}^c ~|~r \downarrow y, \sigma \big]
        \end{eqnarray*}

        Here the second inequality is the induction hypothesis. The third inequality follows because $|D| \geq |F|$ and $\hat{D} \subseteq F$.

        \item $|F| > |D|$: As noted above, this can only happen if $Level(r') = h(u)$ and $w' = v$. Here $x, y \in \isd(w')$. For any node $z \in \isd(w')$ other  than $x$ and $y$, $z \in D \iff z \in F$. As for $x$, we have $x \in F \implies x \in D$. We conclude that $y \in F$ but $y \not\in D$, and thus $|F| = |D| + 1$. As $y$ is not eligible after $\{ r \downarrow y, \sigma \}$, we conclude that after $\{ r \downarrow x, \sigma \}$, the number of servers in $\cS(y)$ with level $\infty$ is exactly one more than the number of $h(u)$-level requests assigned to $y$. Thus, 
        \begin{equation}\label{eq:y-is-great}
            \Pr \big[ {\sf Q}_{x,y} ~|~r \downarrow x, \sigma, r' \downarrow y \big] = 1 
        \end{equation}
        \end{enumerate} 
        \end{enumerate}
        Thus, we have:
        \begin{align}
        \label{keylemma_1}
        \lefteqn{\Pr \big[ {\sf Q}_{x,y} ~|~r \downarrow x, \sigma \big]} \notag\\[0.5em]
          &\geq  \sum_{z \in \hat{D} } \Pr \big[ {\sf Q}_{x,y} ~|~r \downarrow x, \sigma, r' \downarrow z \big]\cdot \Pr \big[r' \downarrow z  ~|~r \downarrow x, \sigma  \big] \notag\\ &\hspace{11.5em} + \Pr \big[ {\sf Q}_{x,y} ~|~r \downarrow x, \sigma, r' \downarrow y \big]\cdot \Pr \big[r' \downarrow y  ~|~r \downarrow x, \sigma  \big] \notag \\[0.5em] 
          &\geq  \sum_{z \in \hat{D} } \Pr \big[ {\sf Q}_{y,x}^c ~|~r \downarrow y, \sigma, r' \downarrow z \big]\cdot \Pr \big[r' \downarrow z  ~|~r \downarrow x, \sigma  \big]  + \Pr \big[r' \downarrow y  ~|~r \downarrow x, \sigma  \big] 
        \end{align} 
        Here the second inequality uses the induction hypothesis and Equation~\ref{eq:y-is-great}. Notice that request $r'$ has exactly $|D|+1$ eligible choices belonging to $F$,
        \begin{align*}
          &Eq.~\ref{keylemma_1}\\
          & =   \sum_{z \in \hat{D} } \Pr \big[ {\sf Q}_{y,x}^c ~|~r \downarrow y, \sigma, r' \downarrow z \big]\cdot \frac{1}{|D| + 1}  + \frac{1}{|D| + 1} \\
          & =   \sum_{z \in \hat{D} }  \Pr \big[ {\sf Q}_{y,x}^c ~|~r \downarrow y, \sigma, r' \downarrow z \big]\cdot \frac{1}{|D|} \\
          &\qquad \qquad - \sum_{z \in \hat{D} }  \Pr \big[ {\sf Q}_{y,x}^c ~|~r \downarrow y, \sigma, r' \downarrow z \big]\biggl( \frac{1}{|D|^2+|D|} \biggr) + \frac{1}{|D| + 1}  \\
          & =  \Pr \big[ {\sf Q}_{y,x}^c ~|~r \downarrow y, \sigma \big] - \sum_{z \in \hat{D} }  \Pr \big[ {\sf Q}_{y,x}^c ~|~r \downarrow y, \sigma, r' \downarrow z \big] \biggl( \frac{1}{|D|^2+|D|} \biggr) + \frac{1}{|D| + 1} \\
          & \geq  \Pr \big[ {\sf Q}_{y,x}^c ~|~r \downarrow y, \sigma \big] - \sum_{z \in \hat{D} }  \biggl( \frac{1}{|D|^2+|D|} \biggr) + \frac{1}{|D| + 1} \\
          & \geq  \Pr \big[ {\sf Q}_{y,x}^c ~|~r \downarrow y, \sigma \big] - \frac{|D|}{|D|^2+|D|} + \frac{1}{|D| + 1} \\
          & =  \Pr \big[ {\sf Q}_{y,x}^c ~|~r \downarrow y, \sigma \big]
        \end{align*}
        This completes the proof of Lemma \ref{lem:key}. \hfill \qed
        \end{proof}
        
\def\SLR{\textsf{SL}'}

 Recall that $v \in \Gamma$ is not a leaf. We can now bound the number of nodes in $\isd(v)$ that $r$ descends into. Let $\textsf{T}_{v}:= | \{w \in \isd(v) \ | \ r \mbox{ descends into } w\}|$. (Recall that $r$ descends into any specific $w \in \isd(v)$ at most once.) The following corollary follows directly from Lemma~\ref{lem:numberofjumpsF}.
 
\begin{corollary}\label{clm:logbound} We have $\E[{\sf T}_{u}] = O(\log \Delta^j) $, where $j \leq 2\gap$. Furthermore, for any $v \in \Gamma \setminus \{u\}$ such that $v$ is not a leaf,  we have $\E[{\sf T}_{v} \mid r \downarrow v ] = O(\log \Delta^\gap) $, where $ \gap = 2 \log \log \Delta$ is the height separation between consecutive layers. 
\end{corollary} 
\begin{proof} We present the proof for $v \in \Gamma \setminus \{u\}$ such that $v$ is not a leaf. The argument for $v = u$ is similar; the only difference is that the height difference between $u$ and nodes in $\isd(u)$ is upper bounded by $2 \gap$ rather than $\gap$. As in the proof of Lemma~\ref{lem:numberofjumpsF}, let 
\begin{equation*} 
S = \{ w \in \isd(v) \ | \ \mbox{$w$ is eligible (for request $r$}) \}
\end{equation*}

Let $|S_0|$ be the size of set $S$ when event ${r \downarrow v}$ happens. We say that the algorithm is in phase $i$ when $|S|$ satisfies,
\[\frac{|S_0|}{2^{i}} < |S| \leq \frac{|S_0|}{2^{i-1}}\]
Phases end when $|S|<1$. There can be at most $\log |S_0|$ phases.

Now, $X_i$ denote number of times $[r \downarrow w, r \uparrow w]$ happens in phase $i$ for some $w \in \isd(v)$. Therefore, ${\sf T}_{v} =X_1+\dots+X_k$. To bound $\E[X_i],$ we consider a particular occurrence of event $[r \downarrow w, r \uparrow w]$, and let $S_{current}$ and $S_{new}$ denote the set $S$ before and after the event.  By Lemma~\ref{lem:numberofjumpsF}, we have
\[\E\Big[|S_{\rm new}|\Big] \le \frac{|S_{\rm current}|}{2}\] 

By Markov's inequality, 
\[ \Pr \left[|S_{\rm new}| \leq \tfrac{2|S_{\mathrm{current}}|}{3} \right] \geq \frac{1}{4}\]

Let us call an occurrence of event $[r \downarrow w, r \uparrow w]$ {\em good} 
if $|S_{\rm new}| \leq \frac{2|S_{\mathrm{current}}|}{3}$. Thus the probability of an event $[r \downarrow w, r \uparrow w]$ being good is at least $1/4$. Furthermore, given the requirements of $|S|$ for phase $i$, there can be at most two good events in phase $i$. Thus, $E[X_i] = O(1)$. We conclude that \[\E[{\sf T}_{v} \mid {r \downarrow v}] = \ \sum_{i=1}^{k} \E[X_i] \le O(1) \cdot \log |S_0| \le O(\log \Delta^\gap)\] \hfill \qed \end{proof}

\def\distToLeaf{{\beta}}

\subsubsection{Aggregating Over All Special Descendants of $u$.}\label{sec:aggregate}

We can now proceed to bound $\E[ \reqcost(r,u) ]$.
\begin{lemma}\label{lemma:requestcost}
    For any internal node $u$ such that request $r$ arrives in subtree $T_u$, \[\E[\reqcost(r,u) \mid r \in \zeta(u)] \leq O( (\log\log \Delta) \cdot  \log \Delta) \cdot \alpha^{h(u)-1}\]
\end{lemma}
\begin{proof} 

For a node $v$ in the HST, we denote by $\distToLeaf(v)$ the distance of $v$ to any leaf in the subtree $T_v$. Thus, $\distToLeaf(v) = 0$ if $v$ is a leaf. Otherwise,
\begin{equation}\label{eq:distToLeaf}
    \distToLeaf(v) = \sum_{j=0}^{h(v) - 1} \alpha^j 
\end{equation}

The first assignment of a server to $r$ after $Level(r)$ equals $h(u)$ happens via $\ALLOCATE(r)$ (in Algorithm~\ref{alg:NewStrategy}) and costs $2d(r,u) = 2 \distToLeaf(u)$. Subsequently, $r$ may get reassigned from a server $s \in A(r, h(u))$ to a server $s' \in A(r,h(u))$ using $\REASSIGN(r,s)$. Recall that the algorithm computes $y$, the first node, from the set consisting of the special descendants of $u$ and $u$ itself, on the path from $s$ to $u$ whose subtree contains an available server (that is, a server in $A(r, h(u))$). The cost of reassigning $r$ from $s$ to $s'$ is $d(s,s') \leq 2d(y,s) = 2 \distToLeaf(y)$. Note that when this reassignment happens, $r$ ascends out of a node in $\isd(y)$ and descends into another node of
$\isd(y)$. Thus, we have 
\begin{equation} \label{eq:reqcostBound}
  \reqcost(r, u)  \leq \sum_{v \in \SpServU} 2 \distToLeaf(v) \cdot {\sf T}_v 
\end{equation}

\paragraph{Number of descends.} Suppose that the special layers containing special descendants of $u$ are $\SL_0, \SL_1, \ldots, \SL_{k-1}$. Recall that $\SL_0$ consists of the leaves in the HST, and that $u$ itself may or may not be in a special layer. For each $1 \leq p \leq k$, let $\SLR_p = \SL_{k-p}$; thus the special layers containing special descendants of $u$, enumerated in order of decreasing height, are $\SLR_1, \SLR_2, \ldots, \SLR_k$. By a slight abuse of notation, let $\SLR_0 = \{u\}$.  Recall that for $1 \leq p \leq k-1$, the height difference between 
$\SLR_p$ and $\SLR_{p+1}$ is $\gap = 2 \log \log \Delta$; the height difference between $\SLR_0$ and $\SLR_1$ is some $j <  2 \gap$.
As $\distToLeaf(v)$ is a constant over all nodes $v \in \SLR_p$, we denote this quantity by $\distToLeaf_p$. 

For each $1 \leq p \leq k$, let 
\begin{equation*}
    \cY_p = | \{ v \in \SLR_{p} \cap \Gamma \ | \ r \mbox{ descends into } v \} |
\end{equation*}
We can now rewrite Inequality~\ref{eq:reqcostBound} as 
\begin{equation} \label{eq:reqcostBound1}
   \reqcost(r, u)  \leq \sum_{v \in \SpServU} 2 \distToLeaf(v) \cdot {\sf T}_v  =
   \sum_{p=1}^k 2 \beta_{p-1} \cY_p
\end{equation}

Thus, we are interested in bounding $\E[\cY_p]$. For $p \geq 2$, we have:
\begin{align*}
    \E[\cY_{p}] &= \sum_{v \in \SL'_{p-1} \cap \Gamma} \E[{\sf T}_{v} \mid {r \downarrow v}] \cdot \Pr[{r \downarrow v}]\\
    &\leq (c_1 \log \Delta^\gap) \cdot \sum_{v \in \SL'_{p-1} \cap \Gamma} \Pr[{r \downarrow v}] \\
    &= (c_1 \log \Delta^\gap) \cdot \E[\cY_{p-1}]  
\end{align*}  
Here the inequality follows from Corollary~\ref{clm:logbound}, and $c_1 > 0$ is the corresponding constant.
Also, $\E[\cY_1]= O(\log \Delta^j)$, where $j \leq 2 \gap$, directly from Corollary~\ref{clm:logbound}. Hence, for any $1\leq p\leq k$ after unrolling the recursion we get 
\begin{equation} \label{eq:Yp}
 \E[\cY_p] \leq O(\log \Delta^j) \cdot (c_1 \log \Delta^\gap)^{p-1}, \mbox{ where $j \leq 2 \gap$} 
\end{equation}

\paragraph{Cost for descends.} From~Eq.~\ref{eq:distToLeaf}, it follows that $\distToLeaf_p \leq \distToLeaf_{p-1}/ \alpha^{\gap}$ for $p \geq 2$, and $\distToLeaf_1 \leq \distToLeaf(u) / \alpha^j$. Here and below it is implicit that $j \leq 2 \gap$.

Applying Expectation over Inequality~\ref{eq:reqcostBound1}, and plugging in the bound from Inequality~\ref{eq:Yp} and the above bounds on $\beta_p$, we get:

\begin{align}\label{eq:costboundsum}
      &\lefteqn{\E[\reqcost(r,u) \mid r \in \zeta(u)]} \nonumber\\ & \leq \sum_{p=1}^{k} 2 \distToLeaf_{p-1} \cdot \E[\cY_p] \nonumber 
      \\ & \le 2 \distToLeaf(u) \cdot O(\log \Delta^j) \left(1 +  \sum_{p=2}^{k} \frac{(c_1 \log \Delta^\gap)^{p-1}}{\alpha^j (\alpha^\gap)^{p-2}} \right) \nonumber \\
      &= 2 \distToLeaf(u) \cdot O(\log \Delta^j) \left(1 + \frac{(c_1 \log \Delta^\gap)}{\alpha^j} +\sum_{p=3}^{k} \frac{(c_1 \log \Delta^\gap)^{p-1}}{\alpha^j (\alpha^\gap)^{p-2}}\right) \nonumber\\
       & \le 2 \distToLeaf(u) \cdot O(\log \Delta^j) \left( 1 + \frac{(c_1 \log \Delta^\gap)}{\alpha^j} \left( 1 +  \underbrace{\sum_{p=3}^{k} \frac{(c_1 \gap \log \Delta)^{p-2}}{(\alpha^\gap)^{p-2}}}_{I} \right) \right)
\end{align}

By setting $\alpha=2$, $\gap=2\log_2\log \Delta$ in~Eq.~\ref{eq:costboundsum}, for the last term in summation, 
\begin{align*}
           I &= \sum_{p=3}^{k} \frac{(c_1 \cdot2 \log_2 \log \Delta \cdot (\log \Delta ) )^{p-2}}{(2^{2 \log_2 \log \Delta})^{p-2}} \\ 
           &= \sum_{p=3}^{k} \frac{(2 c_1 \log_2 \log \Delta)^{p-2} \cdot (\log \Delta)^{p-2}}{(\log \Delta)^{{2}(p-2)}} \\
           &= \sum_{p=3}^{k} \frac{(2 c_1 \log_2 \log \Delta)^{p-2}}{(\log \Delta)^{p-2}}
\end{align*}
The ratio $\frac{2 c_1 \log \log \Delta}{\log \Delta}$ is smaller than $1/2$ for sufficiently large $\Delta$. Therefore, the series converges to some constant $c'$. 

First let us look at the general and complex case where the height $h(u)$ is at least $\gap$. In this case $\gap \leq j < 2 \gap$.
Thus in Eq. \ref{eq:costboundsum}, we have $\frac{ \log \Delta^j}{\alpha^j} = (j\log \Delta)/2^j$, which can be upper bounded by picking $j=\gap$, thus giving $\frac{2 \log \log \Delta \cdot (\log \Delta)}{(\log \Delta)^2}$, which is $O(1)$.

Thus, the expression in Eq.~\ref{eq:costboundsum} is bounded by $O(\beta(u) \cdot \gap \cdot \log \Delta)$ = $\log \log \Delta \cdot (\log \Delta) \cdot \beta(u)$. Noting that $\beta(u) = d(r,u)$, we get,
\begin{eqnarray*}
\E[\reqcost(r,u) \mid r \in \zeta(u)] &\leq O( (\log \log \Delta) \log \Delta) \cdot 2d(r,u) \\ &\leq O(( \log \log \Delta) \log \Delta) \cdot \alpha^{h(u)-1} 
\end{eqnarray*}
Now let us consider the case $h(u) < \gap$. In this case there is only one special layer with special descendants of $u$, and thus $k = 1$. We have,
\begin{eqnarray*}
\E[\reqcost(r, u) \mid r \in \zeta(u)] = O( 2 \beta_u \cdot \log \Delta ^j ) = O(( \log \log \Delta) \cdot \log \Delta) \cdot \alpha^{h(u)-1},
\end{eqnarray*}
as $j < \gap = 2 \log \log \Delta$.
\hfill \qed \end{proof}

\subsection{Conversion to Online Algorithm}\label{sec:online-conv}
The online algorithm runs the reassignment algorithm as requests arrive, and performs the online matching guided by the reassignment algorithm. Both the reassignment algorithm and the online algorithm have the same set of matched and unmatched servers. Suppose that a newly arriving request $r$ gets matched to server $s_1$, which causes a reassignment of request $r_1$ from $s_1$ to $s_2$, and so on. That is, the assignments
\[ s_1 \leftrightarrow r_1, s_2 \leftrightarrow r_2, \ldots, s_t \leftrightarrow r_t \]
that existed prior to the arrival of $r$ are transformed after the arrival of $r$ to
\[ r \leftrightarrow s_1, r_1 \leftrightarrow s_2, \ldots, r_{t-1} \leftrightarrow s_t, r_t \leftrightarrow s_{t+1},\]
where server $s_{t+1}$ was previously unmatched. Then the online algorithm matches $r$ to $s_{t+1}$. For the cost, we note that
\[ d(r,s_{t+1}) \leq d(r,s_1) + \sum_{j=1}^t d(s_j,s_{j+1}),\]
and the RHS is upper bounded by the cost of all the assignments and reassignments triggered by the arrival of $r$. It follows that, even though our way of charging reassignment costs is different from
the BBGN algorithm, the matching cost of the online algorithm is at most that of the reassignment algorithm.

We can now explain why we do not refer to our reassignment model as the ``restricted reassignment model'', the name used by \cite{metriclog2gupta} for their reassignment model. For a request $r$ that is currently matched to server $s$, we do not impose any restriction (in our model!) on the servers that $r$ can be reassigned to. This is unlike \cite{metriclog2gupta}. All we require is that if we reassign $r$ from server $s$ to server $s'$, then we pay a reassignment cost of at least $d(s,s')$. 

\subsection{Competitive Ratio} \label{SS:4.3}

We conclude that the Online Algorithm~\ref{alg:NewStrategy} is $O((\log\log \Delta) \cdot  \log \Delta)$-competitive on a $\Delta$-ary $2$-HST.
This is immediate from Theorem~\ref{thm:comp-in-reassign} and the fact that the conversion from an algorithm in the reassignment model to an online algorithm does not increase the cost, as explained in Appendix~\ref{sec:online-conv}. This completes the proof of Theorem~\ref{thm:main-result}. \hfill \qed
\section{Lower Bound on BBGN's Competitive Ratio}
\label{sec:app-bbgn-lb}
We observe that there is a lower bound of $\Omega(\log n)$ on the expected competitive ratio of the BBGN algorithm \cite{metriclog2gupta} on $2$-ary $2$-HSTs. The example HST consists of a binary tree $T$, rooted at node $x$. The right subtree $T_R$ of $x$ consists of a complete binary tree with $n - 1$ leaves. The left subtree $T_L$ is a root to leaf path ending in a single leaf, which we denote by $y$. The number of edges on every root to leaf path is the same as in any HST; the edge lengths on any root to leaf path decrease by a factor of $2$ like in any $2$-HST.

There is one server located at each leaf of $T$. As for the requests, the first two requests $r_1$ and $r_2$ arrive at the only leaf $y$ in $T_L$. The remaining requests $r_3, r_4, \ldots, r_n$ arrive at distinct leaves in $T_R$ in an arbitrary way. Note that there is one leaf $z$ in $T_R$ where no request arrives. The optimal solution of assigning $r_1$ to the server at $y$, $r_2$ to the server at $z$, and the remaining requests to co-located servers. The cost of the optimal solution is $\beta$, which we define to be twice the distance of the root from a leaf.

Now consider running the BBGN reassignment algorithm on this instance. When $r_2$ arrives, it finds that the colocated server at $y$  is matched to $r_1$. The algorithm therefore assigns $r_2$ to a server in $T_R$ chosen uniformly at random, after setting $Level(r_2)$ to $h(x)$. Subsequent requests are assigned to colocated servers, but when a request arrives at the location where the server assigned to $r_2$ is stationed, $r_2$ is reassigned once again to a unmatched server from $T_R$ chosen uniformly at random. And this process repeats.

Each time the request $r_2$ is reassigned, the BBGN algorithm charges a cost of $\beta$. It is easy to see that the expected number of reassignments of $r_2$ is $\Omega(\log n)$. This shows that the BBGN analysis yields a $\Theta(\log n)$ competitive ratio for their reassignment algorithm. 

Furthermore, we can specialize this example so that, if we consider the cost of the final online matching that is obtained via their reassignment algorithm, the expected competitive ratio is still $\Omega(\log n)$. That is, we can show that one can't really exploit the slack in the way reassignment costs are charged in the BBGN algorithm. For this modified example, we take a random permutation of the leaves in $T_R$, and assign requests $r_3, r_4, \ldots, r_n$ to be the first $n-2$ leaves in the permutation; the requests $r_1$ and $r_2$ still arrive at the only leaf $y$ in $T_L$.

Note that whenever $r_2$ is reassigned from a server $s$ in $T_R$, it is assigned to an available server $s'$ in $T_R$ chosen uniformly at random. The above placement of requests ensures that the expected distance between $s$ and $s'$ is $\Omega(\beta)$; that is, the expected distance from the location of the current request $r_k$ to a location chosen uniformly at random from the locations of the remaining requests $r_{k+1}, r_{k+2}, \ldots, r_n$ and $z$ is $\Omega(\beta)$.

\section{The Approach of Kalyanasundaram et al.}
\label{sec:KPS}
Kalyanasundaram et al.~\cite{KPS23} give a reduction from the online metric matching problem on 2-HSTs to several instances of online matching on the uniform metric. They also give a specific randomized online algorithm for the uniform metric whose expected competitive ratio is $\Theta(\log K)$, where $K$ is the number of distinct server locations. In particular, the competitive ratio is independent of the number $n$ of servers and requests. Here we give an instance on a $\Delta$-ary 2-HST where the expected competitive ratio that follows from their reduction is $\Omega(\log^2 \Delta)$. 

The instance is defined on a $\Delta$-ary $2$-HST of height 2. The root $x$ has $\Delta$ children $\{v_1,\dots,v_\Delta\}$ and each child $v_i$ has $\Delta$ children which are leaf nodes, denoted by $w_{i,1}, w_{i,2}, \dots, w_{i,\Delta}$. There is one server at each leaf, for a total of $\Delta^2$ servers. The length of each edge from $x$ to its children is $2$, and the length of each edge from $v_i$ to its children is $1$. The requests will arrive in the following order: the first two requests $r_1$ and $r_2$ will arrive on $w_{1,1}$ and the remaining $\Delta^2-2$ requests will arrive one-by-one on server $w_{i,j}$ in increasing order of $i$ and $j$, i.e. $r_3$ onward requests will sequentially arrive on locations $(w_{1,2},\dots w_{1,\Delta}, w_{2,1},\dots, w_{2, \Delta}, \dots, w_{\Delta,1},\dots,w_{\Delta,\Delta-1})$. Note that no request arrives on $w_{\Delta,\Delta}$. The optimal matching in this instance matches $r_1$ to the server at $w_{1,1}$, $r_2$ to the server at $w_{\Delta,\Delta}$, and rest of the requests to the servers on their arrival locations. The cost for the optimal matching $\OPT$ is $d(r_2,w_{\Delta,\Delta}) =6$, i.e. matching $r_2$ to $w_{\Delta,\Delta}$ via root $x$ costs $6$ and other matching pairs have cost $0$. 

\cite[Section 2.3]{KPS23} give a black box reduction from any instance on a 2-HST to uniform instances. On the above example, this yields one uniform instance $U_x$ at the root, where the $\Delta + 1$ locations in the uniform metric correspond to $x, v_1, v_2, ...., v_{\Delta}$. Location $x$ has no servers, and each location $v_i$ has $\Delta$ servers. Location $v_1$ receives the first $\Delta+1$ requests, $v_2$ the next $\Delta$ requests, and so on till $v_{\Delta-1}$  which receives $\Delta$ requests, and finally $v_\Delta$ receives the last $\Delta-1$ requests. 

This reduction also yields one uniform instance $U_{v_i}$ for each $v_i$. These instances can be partitioned into three types:

\begin{enumerate}
\item For the instance $U_{v_1}$, it contains $v_1$ and $\{w_{1,1}, \dots, w_{1,\Delta}\}$ as locations, where every location has one server except $v_1$, which has no servers. Location $w_{1,1}$ gets requests $r_1,r_2$, and each $w_{1,j},~j \in\{2,\dots, \Delta - 1\}$, receive one request in that order. Note that in $U_{v_1}$ we do not have the request at $w_{1, \Delta}$; the black-box reduction ``moves'' that request in the original instance to some other instance $U_{v_i}$, where $v_i \in Z$, as described below. Note that the optimum cost for instance $U_{v_1}$ is $2$, the distance in that uniform metric.

\item There is a subset $Z \subseteq \{v_2, v_3, \ldots, v_{\Delta} \}$ where each $U_{v_i}$, for $v_i \in Z$, has $w_{i,1}, w_{i,2}, \dots, w_{i, \Delta}$ and $v_i$ as locations. There is one server at each $w_{i,j}$, and no server at $v_i$. The request sequence has the first request at $v_i$, followed by $\Delta - 1$ other requests, each at a different $w_{i,j}$. Note that the optimum cost for instance $U_{v_i}$ for each $v_i \in Z$ is $2$, the distance in that uniform metric.

\item For each remaining $v_i$ that is not in $Z \cup \{v_1\}$, the uniform instance $U_{v_i}$ has $v_i$ and $w_{i,1}, w_{i,2}, \dots, w_{i, \Delta}$ as locations. There is no server at $v_i$, and one server at each of the $w_{i,j}$. 
The request sequence has $\Delta$ requests arrive at $w_{i,1}, w_{i,2}, \ldots, w_{i, \Delta}$, in that order, just like in the original HST instance. The optimum cost for each such $U_{v_i}$ is $0$, as each request can be matched to the co-located server.
\end{enumerate}

The optimal solution for $U_x$ has cost $4$ -- the distance in that uniform metric. The analysis of \cite{KPS23} upper bounds the expected cost for this uniform metric by $O(\log \Delta)$ here. This only gives an upper bound of $O(\log \Delta)$ on the expected value of $|Z|$. Now, for  each instance $U_{v_i}$ corresponding to $Z$, the key point is that the optimum cost is $2$, and not $0$. Thus the expected cost of the online algorithm of \cite{KPS23} for each such instance $U_{v_i}$ is $O(\log \Delta)$. Thus from their reduction we can only bound the overall expected cost incurred by $E[|Z|] \times O(\log \Delta)$, and thus $O(\log^2 \Delta)$.

This upper bound on the overall expected cost is actually tight, that is, the overall expected cost incurred is $\Omega(\log^2 \Delta)$. This is based on two assertions: (i) $E[|Z|]$ is $\Omega(\log \Delta)$, and (ii) for each instance $U_{v_i}$ corresponding to $Z$, the expected cost of the online algorithm of \cite{KPS23} is $\Omega(\log \Delta)$. These assertions depend on the precise details of the reassignment algorithm of \cite{KPS23} for the uniform metric; this is the algorithm that is then converted to an online algorithm for the uniform metric.

We begin by sketching the argument for (ii). Assume without loss of generality that the requests arrive at $v_i, w_{i,1}, w_{i,2}, \dots, w_{i, \Delta - 1}$ in that order, and that $w_{i, \Delta}$ is the only leaf with no request. The first request that arrives at $v_i$ is the only {\em tourist} in the reassignment algorithm of \cite{KPS23}. When it arrives, it is assigned uniformly at random to one of the $\Delta$ servers. Suppose that server happens to be the one at $w_{i,\Delta}$; then the tourist undergoes no more reassignments. If instead the tourist is assigned to a server $w_{i,j}$ with $j < \Delta$, then when the request at $w_{i,j}$ arrives, the tourist needs to be reassigned. It is reassigned to a server uniformly at random from the ones at $w_{i,j+1}, w_{i,j+2}, \ldots, w_{i, \Delta}$. And the process continues in this way till the tourist is assigned the server at $w_{i, \Delta}$. 

The cost for instance $U_{v_i}$ is proportional to the number of assignments and reassignments of the tourist. It is easy to verify that in expectation this is $\Omega(\log \Delta)$.

Now we address (i) and the instance $U_x$. Here the only tourist in the reassignment algorithm of \cite{KPS23} is the $(\Delta + 1)$'th request to arrive at $v_1$. When this tourist arrives, the algorithm designates 2 servers each at $v_2, \ldots, v_{\Delta}$ as available for the tourist to {\em park} in. It also orders the two servers at each $v_i$; combined with the ordering $v_2, \ldots, v_{\Delta}$ this gives an ordering $\sigma$ of the $2 (\Delta - 1)$ available servers for the tourist. 

Initially the tourist is assigned  a server $s$ chosen uniformly at random from the available servers. If this server is the last one in the ordering $\sigma$, then the tourist undergoes no more reassignments. Otherwise, the tourist does get reassigned at some point, to a server chosen uniformly at random from the ones that occur after $s$ in $\sigma$. And the process continues in this way until the tourist gets assigned to the last server according to $\sigma$. 

It is easy to verify that in $U_{x}$ the number of assignments and reassignments of the tourist is $\Omega(\log \Delta)$ in expectation. Furthermore, $|Z|$ is at least half the number of assignments and reassignments -- the factor of $2$ is because two servers from each $v_i$ were initially designated as available.

\section{Instances with Few Server Locations}\label{sec:leaves}
Consider an online matching instance on an $\alpha$-HST with $k$ leaves. As such an HST is $k$-ary, Theorem~\ref{thm:main-result} yields expected  competitive ratio of $O((\log k) \log \log k)$. With a simple change to the set up, our framework yields the optimal ratio of $O(\log k)$. We set the height difference $\gap$ between consecutive layers to be greater than the height of the tree. This means that we have only one special layer, $\SL_0$, and this layer contains the leaves of the HST. Thus for any internal node $u$, its special descendants are the leaves in the subtree at $u$; these make up the set $\isd(u)$ of $u$'s immediate special descendants also. 

With this change, our reassignment algorithm (Algorithm~\ref{alg:NewStrategy}) boils down to the following. Consider a request $r$ that is about to be assigned a server for the first time after $Level(r)$ attains some value $\ell \geq 0$. If $\ell = 0$, $r$ is assigned an available server at the leaf $\omega$ where $r$ arrived. If $\ell > 0$, then let $u$ denote the ancestor of $r$ at height $\ell$. Then, among the leaves in the subtree of $u$ that have available servers for $r$, we pick one leaf uniformly at random and assign an available server at that leaf to request $r$. Now suppose that at a later time request $r$ is still at level $\ell$, has been assigned a server $s$ at level $\ell$ that resides at a leaf $\beta$ of the subtree at $u$, and now needs to be reassigned. In this case, if $\beta$ has  available servers, we assign one such server to $r$ -- the reassignment cost here is $0$, and this is how the locality of our reassignment strategy manifests. Otherwise, if $u$ has available servers, then, among the leaves in the subtree of $u$ that have available servers for $r$, we pick one leaf uniformly at random and assign an available server at that leaf to request $r$. If $u$ does not have available servers for $r$, we increase $Level(r)$. Thus, our reassignment algorithm (Algorithm~\ref{alg:NewStrategy}) ends up being a natural modification of that of \cite{metriclog2gupta}.

For the analysis, we can show (see Eq.~\ref{eq:reqcost}) that 

\begin{equation}
    \E[\reqcost(r,u) \mid r \in \zeta(u)] \leq c \cdot \log k \cdot \alpha^{h(u) - 1}
\end{equation}
This follows directly from the arguments in Corollary~\ref{clm:logbound}, as $|\isd(u)| \leq k$. As $u$ has no special descendants other than the leaves at the subtree of $u$, there is no need to aggregate costs over other special descendants, as we did in Lemma~\ref{lemma:requestcost}. This gives us the expected competitive ratio of $O(\log k)$.

\section{Other Related Work}\label{sec:otherwork}
Nayyar and Raghvendra \cite{deterministic_line_alg} present an \textit{input-sensitive} analysis of the Robust Matching Algorithm, a deterministic  algorithm introduced in \cite{RM_alg}. They bound the competitive ratio by $\cO(\mu_\M \log^2n)$, where $\mu_\M$ is the ratio of the minimum spanning or TSP tour length to the diameter of the metric space $\M$. Thus they relate the algorithm's competitive ratio to the metric's structural parameter $\mu_\M$, explaining why the algorithm performs significantly better on structured metrics despite having a $\Theta(n)$ worst-case bound. For the line metric $\mu_\M=2$, hence they obtain a $\cO(\log^2 n)$-competitive deterministic line algorithm. For the line metric, \cite{det_line_analysis} improved the upper bound on the competitive ratio for the robust matching algorithm to $\Theta(\log n)$.  Peserico and Scquizzato \cite{line_lowerbound} show a lower bound of $\Omega(\sqrt{\log n})$ on the line metric for randomized algorithms, even with an oblivious adversary. Other useful lower bounds on the competitive ratio for certain important types of algorithms are presented in \cite{collectionlowerboundsonline}.

The competitive ratio under the \textit{adversarial model} with irrevocable decisions is often viewed as overly pessimistic, giving all the power to the adversary. This limitation has motivated alternative models that relax these assumptions, including the frameworks that allow limited recourse and stochastic arrivals. 

A model that is very closely related to the reassignment framework is the \textit{recourse} model introduced by Gupta et al. in~\cite{gupta_et_al:LIPIcs.APPROX/RANDOM.2020.40}. They ask the question, that given an ability to rematch the limited number of previously paired requests and servers, can we obtain a strictly better performance? (unlike the algorithms in the reassignment framework, here one doesn't need to pay the reassignment cost.) Giving an affirmative answer to the question, they achieve $O(\log n)$-competitive ratio with amortized recourse budget of $O(\log n)$. They extend the idea of recourse to dynamic online matching, where the requests and servers can also depart. 

This model has also been studied independently by Megow and N\"olke~\cite{Megow2025}, where they use the different $t$-net cost algorithm by Raghvendra (\cite{RM_alg},~\cite{det_line_analysis}) to achieve $O(1)$-competitive ratio with amortized recourse budget of $O(\log n)$ on the line metric. 

Alternatively, the \textit{stochastic} models have been extensively studied in the literature (see, e.g.,~\cite{randomarrival},~\cite{RM_alg},~\cite{gupta_et_al:LIPIcs.ICALP.2019.67},~\cite{yang2026online},~\cite{smoothed_analysis_delta-ary}). 
There are two popular models in the stochastic setting, 
\begin{enumerate}
    \item The \textit{random-order} arrival: In this, the request sequence is chosen uniformly at random from the multiset of permutations of the locations selected by an adversary.
    \item The \textit{i.i.d.} arrival: We assume the requests arrive as an i.i.d. sequence sampled from a fixed unknown or known distribution.
\end{enumerate} 
%
%
%
\end{document}